\documentclass[runningheads]{llncs}
\usepackage{graphicx}
\usepackage{amsmath}
\usepackage{amssymb}
\usepackage{xspace}
\usepackage{geometry}
\usepackage[font=small, labelfont=bf]{caption}
\usepackage{subfigure}
\usepackage{hyperref}
\usepackage{wrapfig}
\usepackage{thm-restate}
\usepackage{todonotes}
\usepackage{paralist} 
\usepackage[ruled,lined,boxed,linesnumbered]{algorithm2e}

\spnewtheorem*{sketch}{Sketch of proof}{\itshape}{\rmfamily}

\newcommand{\DENSE}{%
\setlength{\labelwidth}{12pt}%
\setlength{\labelsep}{2pt}%
\setlength{\leftmargin}{\labelwidth}%
\addtolength{\leftmargin}{\labelsep}%
\setlength{\parsep}{0pt}%
\setlength{\itemsep}{0pt}%
\setlength{\topsep}{0pt}
}

\hypersetup{pdfborder={0 0 0.4}}

\newcommand{\remove}[1]{}
\newcommand{\euristicaUno}{\texttt{Shorten\-Host\-Switch}\xspace}
\newcommand{\euristicaDue}{\texttt{Search\-Maximal\-Planar}\xspace}
\newcommand{\algoritmoPlanare}{\texttt{PlanarDraw}\xspace}

\title{Visualizing Co-Phylogenetic Reconciliations%
\thanks{This paper appears in the Proceedings of the 25th International Symposium on Graph Drawing and Network Visualization (GD 2017). Please, refer to~\cite{cdmp-vcpr-17}. 
This research was partially supported by MIUR project ``MODE -- MOrphing graph Drawings Efficiently'', prot. 20157EFM5C\_001 and by {\em Sapienza} University of Rome project ``Combinatorial structures and algorithms for problems in co-phylogeny''.}}

\authorrunning{Calamoneri et al.}

\author{Tiziana Calamoneri \inst1 \and Valentino Di Donato \inst2 \and \\ 
        Diego Mariottini \inst2 \and Maurizio~Patrignani \inst2}

\institute{Computer Science Department, University of Rome ``Sapienza'', Rome, Italy \and 
  Engineering Department, Roma Tre University, Rome, Italy}

\graphicspath{{figures/}}

\begin{document}

\maketitle

\begin{abstract}
We introduce a hybrid metaphor for the visualization of the reconciliations of co-phylogenetic trees, that are mappings among the nodes of two trees. 
The typical application is the visualization of the co-evolution of hosts and parasites in biology.
Our strategy combines a space-filling and a node-link approach. Differently from traditional methods, it guarantees an unambiguous and `downward' representation whenever the reconciliation is time-consistent (i.e., meaningful). 
We address the problem of the minimization of the number of crossings in the representation, by giving a characterization of planar instances and by establishing the complexity of the problem.
Finally, we propose heuristics for computing representations with few crossings. 
\end{abstract}


\section{Introduction}

Producing readable and compact representations of trees has a long tradition in the graph drawing research field. In addition to the standard node-link diagrams, which include layered trees, radial trees, hv-drawings, etc., trees can be visualized via the so-called space-filling metaphors, which include circular and rectangular treemaps, sunbursts, icicles, sunrays, icerays, etc.~\cite{r-tda-13,s-ttvr-11}.

Unambiguous and effective representation of co-phylogenetic trees, that are pairs of phylogenetic trees with a mapping among their nodes, is needed in biological research.
A \emph{phylogenetic tree} is a full rooted binary tree (each node has zero or two children) representing the\remove{inferred} evolutionary relationships among related organisms. 
Biologists who study the co-evolution of species, such as hosts and parasites, start with a host phylogenetic tree $H$, a parasite tree $P$, and a mapping function~$\varphi$ (not necessarily injective nor surjective) from the leaves of~$P$ to the leaves of~$H$. 
The triple $\langle H,P, \varphi\rangle$, called \emph{co-phylogenetic tree}, is traditionally represented with a tanglegram drawing, that consists of a pair of plane trees whose leaves are connected by straight-line edges~\cite{bcef-gbtaa-09,bhtw-ffpad-09,bbbnosw-dcbt-12,ds-oloth-04,fkp-ctvcm-10,nvwh-dbtee-09,szh-trptn-11}. However, a tanglegram only represents the input of a more complex process that aims at computing a mapping $\gamma$, called {\em reconciliation}, that extends $\varphi$ and maps all the parasite nodes onto the host nodes.  

Given $H$, $P$, and $\varphi$, a great number of different reconciliations are possible. Some of them can be discarded, since they are not consistent with time (i.e. they induce contradictory constraints on the periods of existence of the species associated to internal nodes).
The remaining reconciliations are generally ranked based on some quality measure and only the optimal ones are considered.
Even so, optimal reconciliations are so many that biologists have to perform a painstaking manual inspection to select those that are more compatible with their understanding of the evolutionary phenomena.

In this paper we propose a new and unambiguous metaphor to represent reconciliations of co-phylogenetic trees (Section~\ref{se:model}). The main idea is that of representing $H$ in a suitable space-filling style and of using a traditional node-link style to represent~$P$. 
This is the first representation guaranteeing the downwardness of~$P$ when time-consistent (i.e., meaningful) reconciliations are considered. 
In order to pursue readability, we study the number of crossings that are introduced in the drawing of tree $P$ (tree $H$ is always planar): on the one hand, in Section~\ref{se:planarity} we characterize planar reconciliations, 
on the other hand, we show in Section~\ref{se:complexity} that reducing the number of crossings in the representation of the reconciliations is NP-complete. 
Finally, we propose heuristics to produce drawings 
with few crossings (Section \ref{se:heuristics}) and experimentally show their effectiveness and efficiency (Section \ref{se:experiments}). Details and full proofs can be found in the appendix.


\section{Background}\label{se:preliminaries}\label{se:related}

In this paper, whenever we mention a tree $T$, we implicitly assume that it is a {\em full rooted binary tree} with node set $\mathcal{V}(T)$ and arc set $\mathcal{A}(T)$, and that arcs are oriented away from the root $r(T)$ down to the set of leaves $\mathcal{V}_L(T) \subset \mathcal{V}(T)$ (see Appendix~\ref{ap:tree-defs} for formal definitions).
The \emph{lowest common ancestor} of two nodes $u,v \in \mathcal{V}(T)$, denoted $lca(u,v)$, is the last 
common node of the two directed paths leading from $r(T)$ to $u$ and $v$. 
Two nodes $u$ and $v$ are \emph{comparable} if $lca(u,v) \in \{u,v\}$, otherwise they are \emph{incomparable}.   


A \emph{tanglegram} $\langle T_1,T_2,G \rangle$ generalizes a co-phylogenetic tree and consists of two generic rooted trees $T_1$ and $T_2$ and a bipartite graph $G=(\mathcal{V}_L(T_1),\mathcal{V}_L(T_2),E)$ among their leaves. In a \emph{tanglegram drawing} of $\langle T_1,T_2,G \rangle$: (i) tree $T_1$ is planarly drawn above a horizontal line $l_1$ with its arcs pointing downward and its leaves on $l_1$, (ii) tree $T_2$ is planarly drawn below a horizontal line $l_2$, parallel to $l_1$, with its arcs pointing upward and its leaves on $l_2$, and (iii) edges of $G$, called \emph{tangles}, are straight-line segments drawn in the horizontal stripe bounded by $l_1$ and $l_2$ (see Fig.~\ref{fig:tanglegram}(a) in Appendix). A decade-old literature is devoted to tanglegram drawings (see e.g.~\cite{bcef-gbtaa-09,bhtw-ffpad-09,bbbnosw-dcbt-12,ds-oloth-04,fkp-ctvcm-10,nvwh-dbtee-09,szh-trptn-11}).
Finding a tanglegram drawing that minimizes the number of crossings among the edges in $E$ is known to be NP-complete, even if the trees are binary trees or if the graph $G$ is a matching~\cite{fkp-ctvcm-10}. 


A 
\emph{reconciliation} of the co-phylogenetic tree $\langle H, P, \varphi \rangle$ is a mapping 
$\gamma: \mathcal{V}(P) \rightarrow \mathcal{V}(H)$ that satisfies the following properties:
(i) for any $p \in \mathcal{V}_L(P)$, $\gamma(p) = \varphi(p)$, that is, $\gamma$ extends $\varphi$,
(ii) for any arc $(p_i,p_j) \in \mathcal{A}(P)$, $lca(\gamma(p_i),\gamma(p_j))) \neq \gamma(p_j)$, that is, a child 
$p_j$ of $p_i$ cannot be mapped to an ancestor of $\gamma(p_i)$ and
(iii) for any $p \in V \setminus \mathcal{V}_L(P)$ with children $p_1$ and $p_2$, $lca(\gamma(p),\gamma(p_1)) = \gamma(p)$ or $lca(\gamma(p),\gamma(p_2))=\gamma(p)$, that is, at least one of the two children is mapped in the subtree rooted at~$\gamma(p)$. 

The set of all reconciliations of $\langle H, P, \varphi \rangle$ is denoted $\mathcal{R}(H, P, \varphi)$.

Four types of events may take place in a reconciliation (see formal definitions in Appendix~\ref{ap:events}): {\em co-speciation}, when both the host and the parasite speciate; {\em duplication}, when the parasite speciates (but not the host) and both parasite children remain associated with the host; {\em loss}, when the host speciates but not the parasite, leading to the loss of the parasite in one of the two host children; and {\em host-switch}, when the parasite speciates and one child remains with the current host while the other child jumps to an incomparable host.

Each of the above events is usually associated with a penalty and 
the minimum cost reconciliations are searched (they can be computed with polynomial delay~\cite{Dal15,CORE-ILP}).
%
However, only reconciliations not violating obvious temporal constraints are of interest. 
A reconciliation $\gamma$ is \emph{time-consistent} if there exists a linear ordering $\pi$ of the parasites $\mathcal{V}(P)$ such that:\begin{inparaenum}[(i)]
\item for each arc $(p_1, p_2) \in \mathcal{A}(P)$, $\pi(p_1) < \pi(p_2)$; 
\item for each pair $p_1, p_2 \in \mathcal{V}(P)$ such that $\pi(p_1) < \pi(p_2)$, $\gamma(p_2)$ is not a proper ancestor of $\gamma(p_1)$.
\end{inparaenum}
%
%
Recognizing time-consistent reconciliations is a polynomial task~\cite{slxsvd-idlti-12,Dal15,CORE-ILP}, while producing exclusively time-consistent reconciliations is NP-complete~\cite{ofcl-cprpn-11,thl-sidlg-11}. This is why usually time-inconsistent reconciliations are filtered out in a post-processing step~\cite{eucalypt}. 


The available tools to compute reconciliations adopt three main conventions to represent them.
%
The simplest strategy, schematically represented in Fig.~\ref{fig:strategies}(a), represents the two trees by adopting the traditional node-link metaphor, where the nodes of~$P$ are drawn close to the nodes of $H$ they are associated to. 
Unfortunately, when several parasite nodes are associated to the same host node, the drawing becomes cluttered and the attribution of parasite nodes to host nodes becomes unclear. Further, even if $P$ was drawn without crossings (tree $H$ always is), the overlapping of the two trees produces a high number of crossings (see, for example, Figs.~\ref{fig:core-pa} and~\ref{fig:jane-4} in Appendix~\ref{ap:tools}).

\begin{figure}[tb]
\centering
\begin{tabular}{c @{\hspace{2em}} c @{\hspace{2em}} c}
\includegraphics[width=3.5cm]{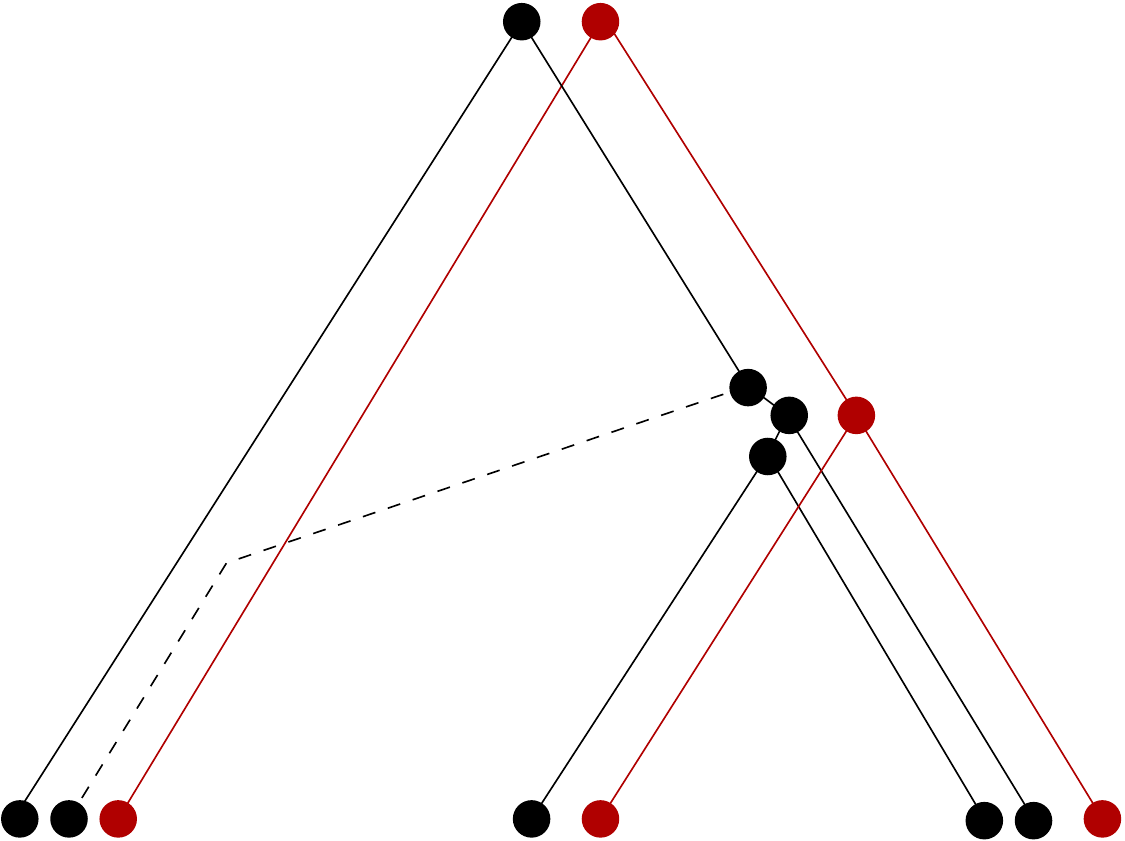} &
\includegraphics[width=3.5cm]{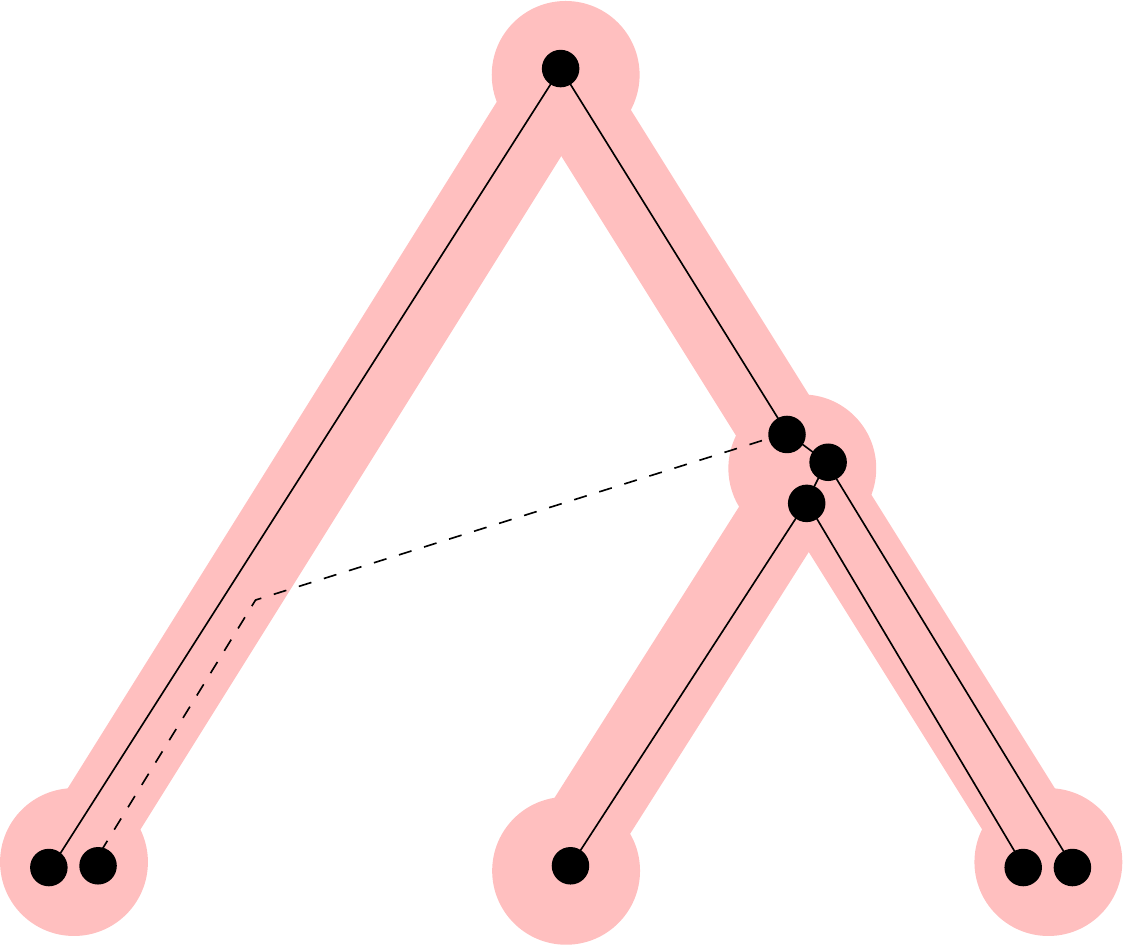} &
\includegraphics[width=3.5cm]{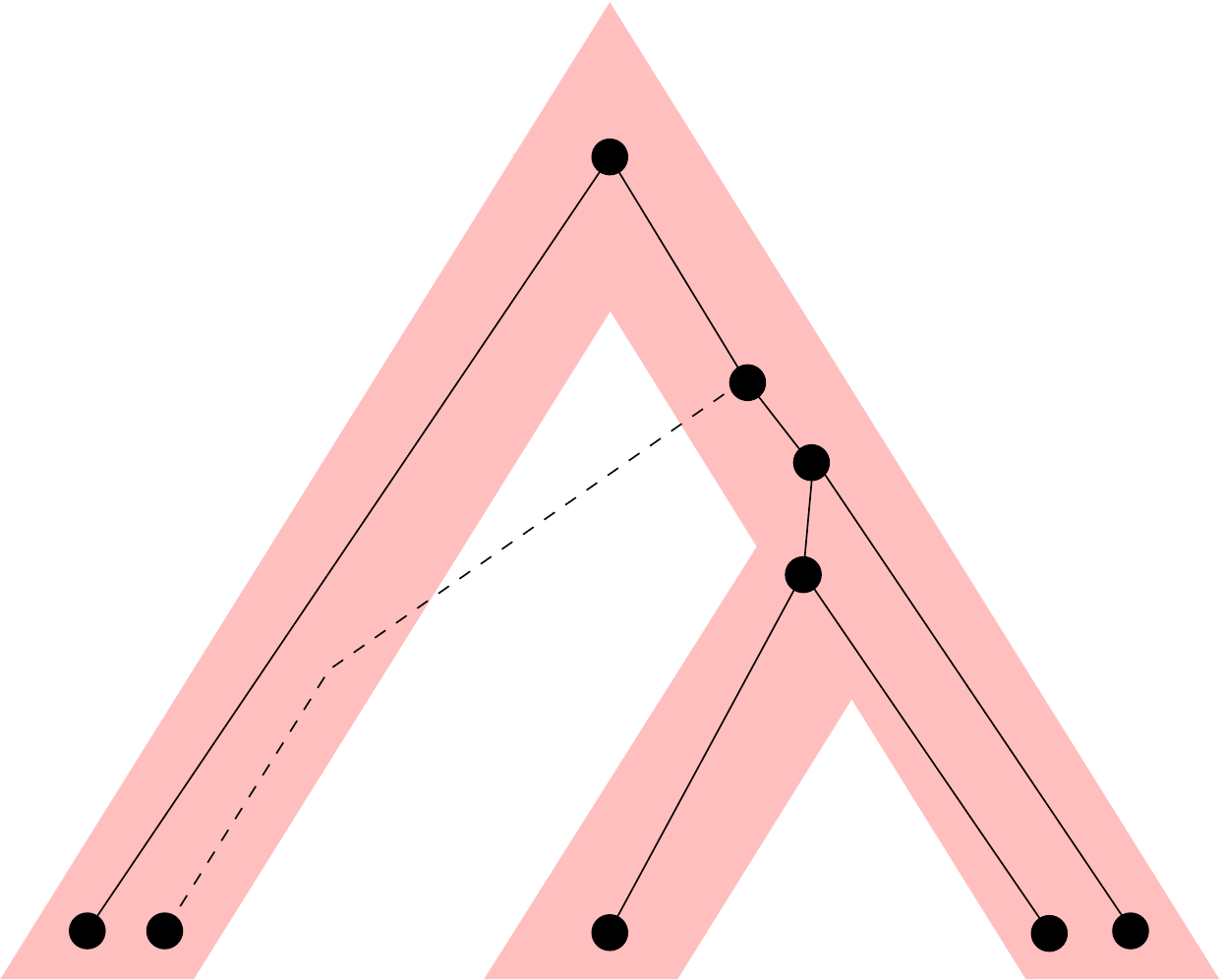} \\
(a) & (b) & (c)
\end{tabular}
\caption{Three visualization strategies for representing co-phylogenetic trees.}\label{fig:strategies}
\end{figure}


An alternative strategy (Fig.~\ref{fig:strategies}(b)) consists in representing $H$ as a background shape, such that its nodes are shaded disks and its arcs are thick pipes, while $P$ is contained in~$H$ and drawn in the traditional node-link style. This strategy is used, for example, by {\em CophyTrees}~\cite{eucalypt}, the viewer associated with the Eucalypt tool~\cite{Dal15}. 
The representation is particularly effective, as it is unambiguous and crossings between the two trees are strongly reduced, but it is still cluttered when a parasite subtree has to be squeezed inside the reduced area of a host node (see Fig.~\ref{fig:eucalypt} in Appendix~\ref{ap:tools}).

Finally, some visualization tools adopt the strategy of keeping the containment metaphor while only drawing thick arcs of $H$ and omitting host nodes (Fig.~\ref{fig:strategies}(c)). This produces a node-link drawing of the parasite tree drawn inside the pipes representing the host tree. Examples include {\em Primetv}~\cite{primetv} and {\em SylvX}~\cite{Sylvx}---see Figs.~\ref{fig:primetv} and~\ref{fig:sylvx} in Appendix~\ref{ap:tools}, respectively. Also this strategy is sometimes ambiguous, since it is unclear how to attribute parasites to hosts.



\section{A new model for the visualization of reconciliations}\label{se:model}

%

%
Inspired by recent proposals of adopting space-filling techniques to represent biological networks~\cite{tk-avpn-16}, and with the aim of overcoming the limitations of existing visualization strategies, we introduce a new hybrid metaphor for the representation of reconciliations. A space-filling approach is used to represent $H$, while tree $P$ maintains the traditional node-link representation. The reconciliation is unambiguously conveyed by placing parasite nodes inside the regions associated with the hosts they are mapped to.
%


\begin{figure}[tb]
\centering
\begin{tabular}{c @{\hspace{2em}} c @{\hspace{2em}} c}
\includegraphics[width=5.5cm]{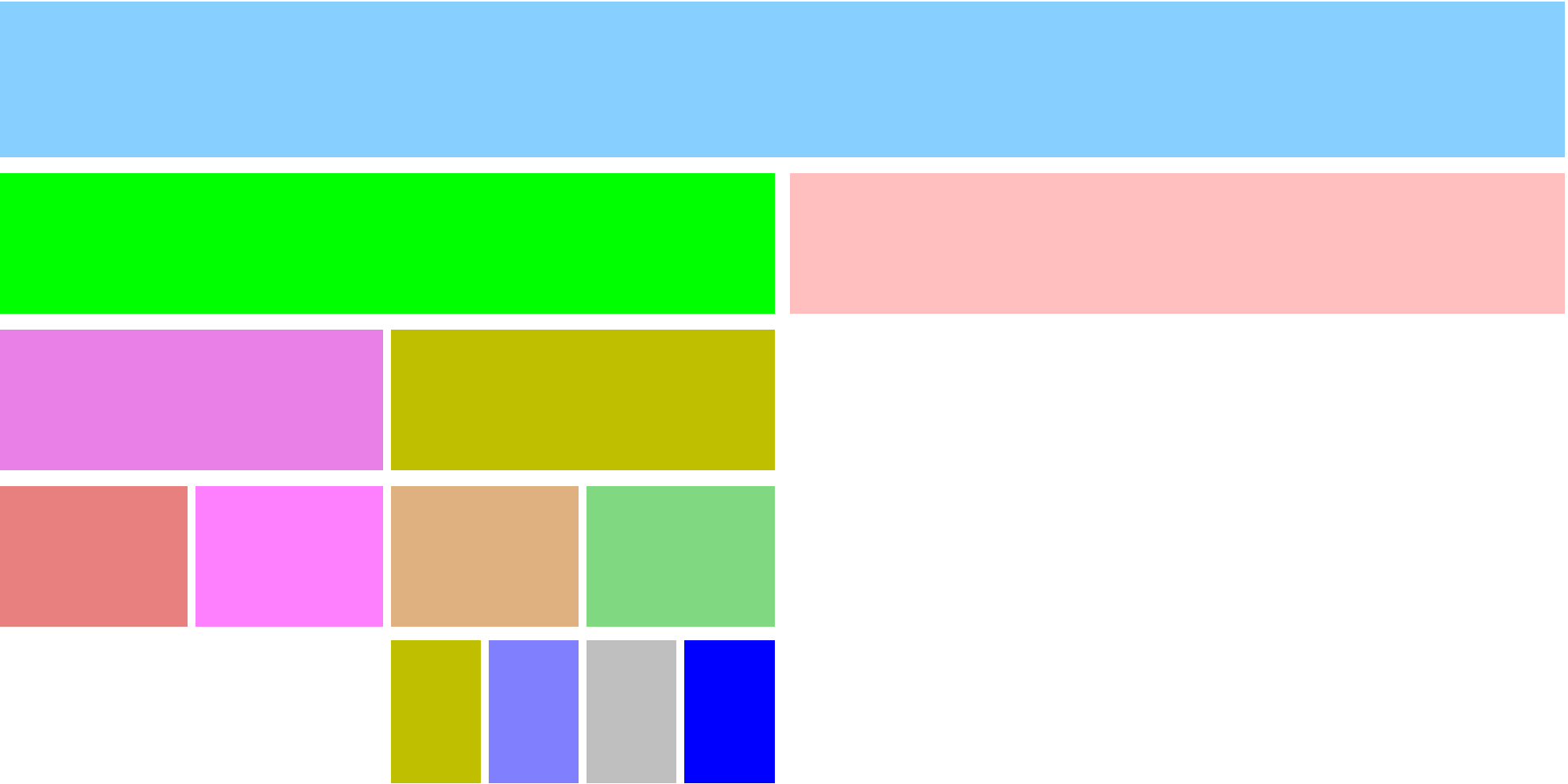} &
\includegraphics[width=5.5cm]{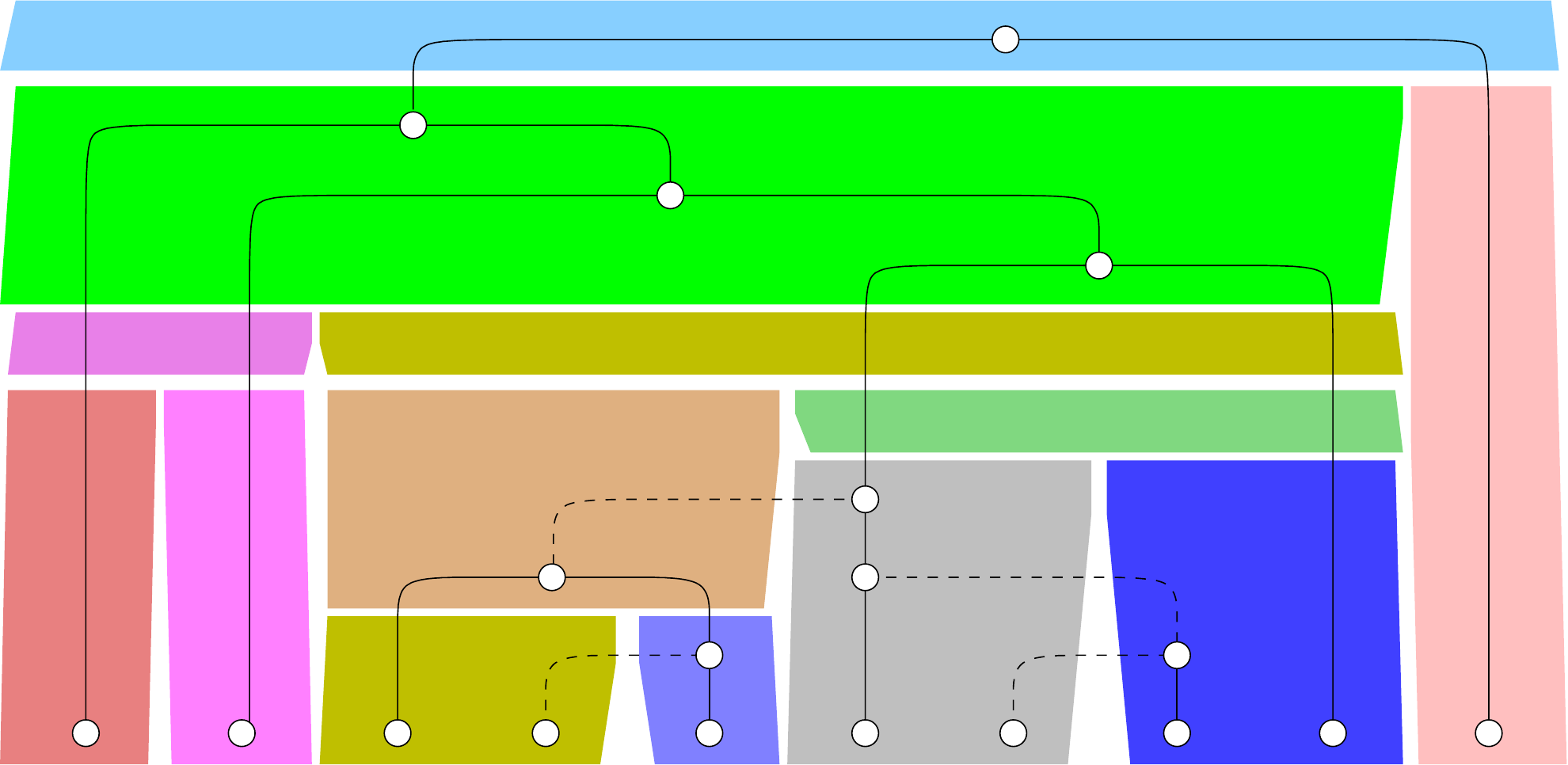} & \\
(a) & (b) 
\end{tabular}
\caption{(a) An icicle. (b) The representation adopted for trees $H$ and $P$.}\label{fig:icicle}
\end{figure}

More specifically, the representation of tree $H$ is a variant of a representation known in the literature with the name of \emph{icicle}~\cite{kl-ipbdh-83}. An {\em icicle} is a space-filling representation of hierarchical information in which nodes are represented by rectangles and arcs are represented by the contact of rectangles, such that the bottom side of the rectangle representing a node touches the top sides of the rectangles representing its children (see Fig.~\ref{fig:icicle}(a)). 
In our model, in order to contain parasite subtrees of different depths, we allow rectangles of different height. Also we force all leaves of $H$ (i.e.\ present-day hosts) to share the same bottom line that intuitively represents current time.


%
%
%
Formally, an {\em HP-drawing $\Gamma(\gamma)$ of $\gamma \in \mathcal{R}(H, P, \varphi)$} is the simultaneous representation of $H$ and $P$ as follows. Tree $H$ is represented in a space-filling fashion such that:
\begin{inparaenum}[(1)]
\item nodes of $H$ are represented by internally disjoint rectangles that cover the drawing area; 
the rectangle corresponding to the root of $H$ covers the top border of the drawing area while the rectangles corresponding to the leaves of $H$ touch the bottom border of the drawing area with their bottom sides; and 
%
\item arcs of $H$ are represented by the vertical contact of rectangles, the upper rectangle being the parent and the lower rectangle being the child.
%
\end{inparaenum}
%
Conversely, tree $P$ is represented in a 
node-link style such that:
\begin{inparaenum}[(1)]
\item each node $p \in \mathcal{V}(P)$ is drawn as a point in the plane inside the representation of the rectangle corresponding to node $\gamma(p)$; and 
\item each arc $(p_1,p_2) \in \mathcal{A}(P)$ is drawn as a vertical segment if $p_1$ and $p_2$ have the same $x$-coordinate; otherwise, it is drawn as a 
horizontal segment followed by a vertical segment.
%
\end{inparaenum}

It can be assumed that an HP-drawing only uses integer coordinates. In particular the corners of the rectangles representing the nodes of $H$ could exclusively use even coordinates and the nodes of $P$ could exclusively use odd coordinates.

Graphically, since the icicle represents a binary tree, we give the rectangles a slanted shape in order to ease the visual recognition of the two children of each node (see Fig.~\ref{fig:icicle}(b)). Also, the bend 
of an arc of $P$ is a small circular arc.


We say that HP-drawing $\Gamma(\gamma)$ is \emph{planar} if no pair of arcs of $P$ intersect except, possibly, at a common endpoint, and that  
it is \emph{downward} if, for each arc $(p_1,p_2) \in \mathcal{A}(P)$, parasite $p_1$ has a $y$-coordinate greater than that of parasite~$p_2$. 

\section{Planar instances and reconciliations}\label{se:planarity}

In this section we characterize the reconciliations that can be planarly drawn, 
showing that a time-consistent reconciliation is planar if and only if the corresponding co-phylogenetic tree admits a planar tanglegram drawing.  

\begin{restatable}{theorem}{theoremplanar}\label{th:planar}
Given a co-phylogenetic tree $\langle H,P, \varphi\rangle$, the following statements are equivalent:
\begin{inparaenum}[(1)]
\item $\langle H,P, \varphi\rangle$ admits a planar tanglegram drawing $\Delta$. 
\item Every time-consistent reconciliation $\gamma \in \mathcal{R}(H,P, \varphi)$ admits a planar downward HP-drawing~$\Gamma(\gamma)$.
\end{inparaenum}
\end{restatable}

\begin{sketch}
First, we prove that $(2)$ implies $(1)$.
Consider a planar drawing $\Gamma(\gamma)$ of $\gamma \in \mathcal{R}(H,P, \varphi)$ and let $l$ be the horizontal line passing through the bottom border of $\Gamma(\gamma)$. Observe that the leaves of $P$ lie above $l$.
Construct a tanglegram drawing $\Delta$ of $\langle H,P, \varphi\rangle$ as follows:
\begin{inparaenum}[(a)]%
\item Draw $H$ by placing each node $h \in \mathcal{V}(H)$ in the center of the rectangle representing $h$ in $\Gamma(\gamma)$ and by representing each arc $a \in \mathcal{A}(H)$ as a suitable curve between its incident nodes;
\item draw $P$ in $\Delta$ as a mirrored drawing with respect to $l$ of the drawing of $P$ in $\Gamma(\gamma)$;
\item connect each leaf $p \in L(P)$ to the host $\gamma(p)$ with a straight-line segment.
\end{inparaenum}
It is immediate that $\Delta$ is a tanglegram drawing of $\langle H,P,\varphi\rangle$ and that it is planar whenever $\Gamma(\gamma)$ is. 


Proving that $(1)$ implies $(2)$ is more laborious. 
Let $\Delta$ be a planar tanglegram drawing of $\langle H,P, \varphi\rangle$ (Fig.~\ref{fig:tanglegram}(a)). We construct a drawing $\Gamma(\gamma)$ of the given time-consistent reconciliation $\gamma \in \mathcal{R}(H,P, \varphi)$ as follows.
%
First, insert into the arcs of $P$ dummy nodes of degree two to represent losses, obtaining a new tree $P'$ (Fig.~\ref{fig:tanglegram}(b)). 
Since $\gamma$ is time-consistent, consider any ordering $\pi'$ of $\mathcal{V}(P')$ consistent with $H$. Remove from $\pi'$ the leaves of $P$ and renumber the remaining nodes obtaining a new ordering $\pi$ from $1$ to $|\mathcal{V}(P')-\mathcal{V}_L(P')|$. 
%
Regarding $y$-coordinates: all the leaves of $P'$ have $y$-coordinate $1$, that is, they are placed at the bottom of the drawing, while each internal node $p \in \mathcal{V}(P')\setminus\mathcal{V}_L(P')$ has $y$-coordinate $2\pi(p)+1$ (see Fig.~\ref{fig:downward}(a)).
Regarding $x$-coordinates: each leaf $p \in \mathcal{V}_L(P)$ has $x$-coordinate $2{\sigma}(p)+1$, where ${\sigma}(p)$ is the left-to-right order of the leaves of $T_2$ in $\Delta$. The $x$-coordinate of an internal node $p$ of $P$ is copied from one of its children $p_1$ or $p_2$, arbitrarily chosen if none of them is connected by a host-switch, the one (always present) that is not connected by a host-switch otherwise. 

Let $h$ be a node of $\mathcal{V}(H)$; 
rectangle $R_h$, representing $h$ in $\Gamma$, has the minimum width that is sufficient to span all the parasites contained in the subtree $T_h(H)$ of $H$ rooted at $h$ (hence, it spans the interval $[x_{min}-1,x_{max}+1]$, where $x_{min}$ and $x_{max}$ are the minimum and maximum $x$-coordinates of a parasite contained in $T_h(H)$, respectively). The top border of $R_h$ has $y$-coordinate $y_{\textsc{min}}-1$, where $y_{\textsc{min}}$ is the minimum $y$-coordinate of a parasite node contained in the parent of $h$. The bottom border of $R_h$ is $y_{min}-1$, where $y_{min}$ is the minimum $y$-coordinate of a parasite node contained in $h$ (see Fig.~\ref{fig:downward}(b)).
%

The proof concludes by showing that the obtained representation $\Gamma(\gamma)$ is planar and downward (See Appendix~\ref{ap:theorem-planar}). \qed
\end{sketch}

We remark that a statement analogous to the one of Theorem~\ref{th:planar} can be proved also for the visualization strategy schematically represented 
in Fig.~\ref{fig:strategies}(b) and adopted, for example, by {\em CophyTrees}~\cite{eucalypt}.

The algorithm we actually implemented, called~\algoritmoPlanare, is a refinement of the one described in the proof of Theorem~\ref{th:planar}. It assigns to the parent parasite an $x$-coordinate that is intermediate between those of the children whenever both children are not host-switches and it produces a more compact representation with respect to the $y$-axis (see Figs.~\ref{fig:icicle}(b) and~\ref{fig:pmp}). 

\begin{figure}[tb]
\centering
\includegraphics[width=12cm]{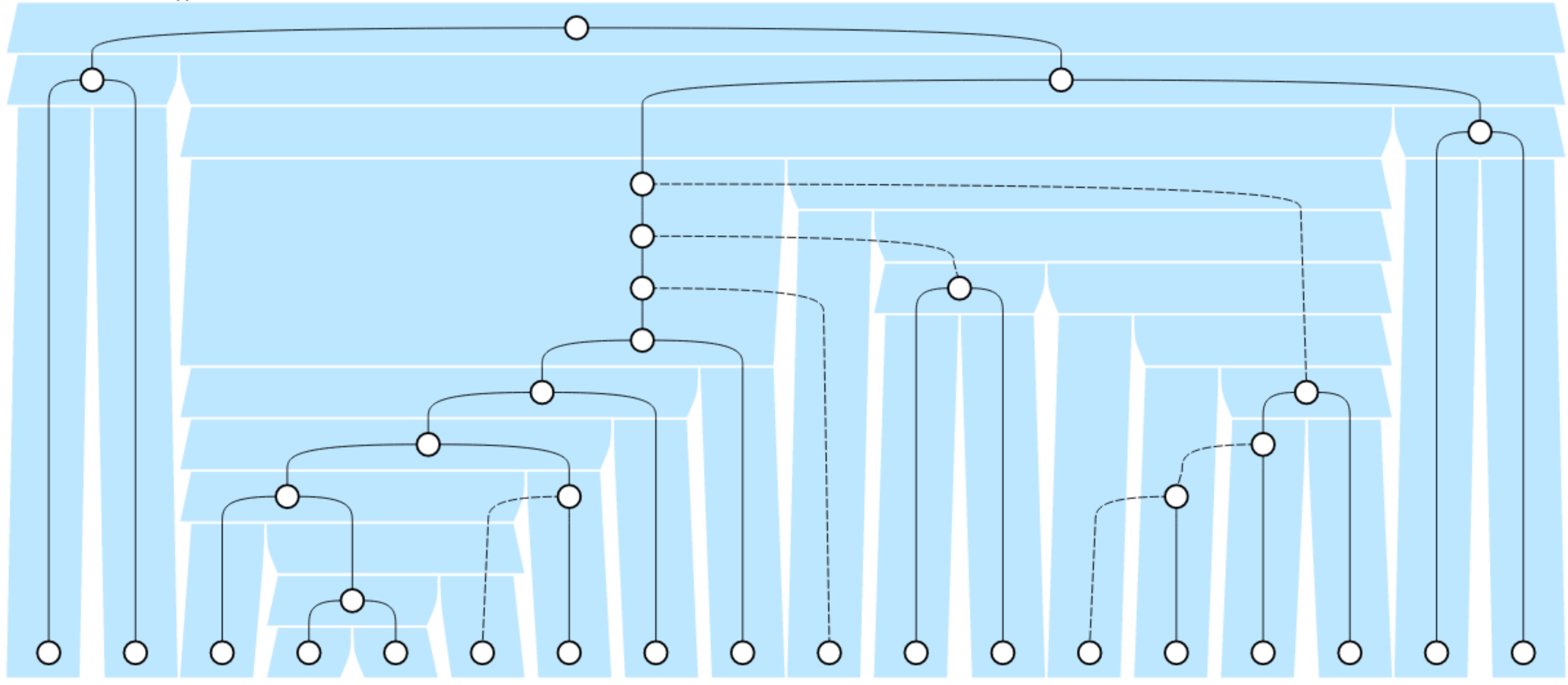}
\caption{A planar HP-drawing of a reconciliation of the co-phylogenetic tree of Pelican \& Lice (MP) computed by Algorithm \algoritmoPlanare.}
\label{fig:pmp}
\end{figure}

\section{Minimizing the number of crossings}\label{se:complexity}

%
In this section we focus on non-planar instances and prove that computing an HP-drawing of a reconciliation with the minimum number of crossings is NP-complete.
%
Given a reconciliation $\gamma \in \mathcal{R}(H, P, \varphi)$ and a constant $k$, we consider the decision problem {\sc Reconciliation Layout (RL)} that asks whether there exists an HP-drawing of $\gamma$ that has at most $k$ crossings.
We prove that \textsc{RL} is NP-hard 
by reducing to it the NP-complete problem {\sc Two-Trees Crossing Minimization (\textsc{TTCM})}~\cite{fkp-ctvcm-10}. 
The input of \textsc{TTCM} 
consists of two binary trees $T_1$ and $T_2$, whose leaf sets are in one-to-one correspondence, and a constant $k$. The question is whether $T_1$ and $T_2$ admit a tanglegram drawing with at most $k$ crossings among the tangles.  
%
In~\cite{bbbnosw-dcbt-12} it is shown that \textsc{TTCM} remains NP-complete even if the input trees are two complete binary trees of height $h$ (hence, with $2^{h}$ leaves). We reduce this latter variant 
to \textsc{RL}.

\begin{restatable}{theorem}{theoremcomplexity}\label{th:complexity}
Problem {\sc RL} is NP-complete.
\end{restatable}

\begin{sketch}
Problem \textsc{RL} is in NP by exploring all possible HP-drawings of~$\gamma$.
Let $I_{\textsc{TTCM}} = \langle T_1, T_2, \psi, k \rangle$ be an instance of \textsc{TTCM}, where $T_1$ and $T_2$ are complete binary trees of height $h$, $\psi$ is a one-to-one mapping between $\mathcal{V}_L(T_1)$ and $\mathcal{V}_L(T_2)$, and $k$ is a constant. We show how to build an equivalent instance $I_{\textsc{RL}} =\langle \gamma \in \mathcal{R}(H, P,\varphi), k' \rangle$ of \textsc{RL}. 

\begin{figure}[tb]
\centering
\includegraphics[width=11cm]{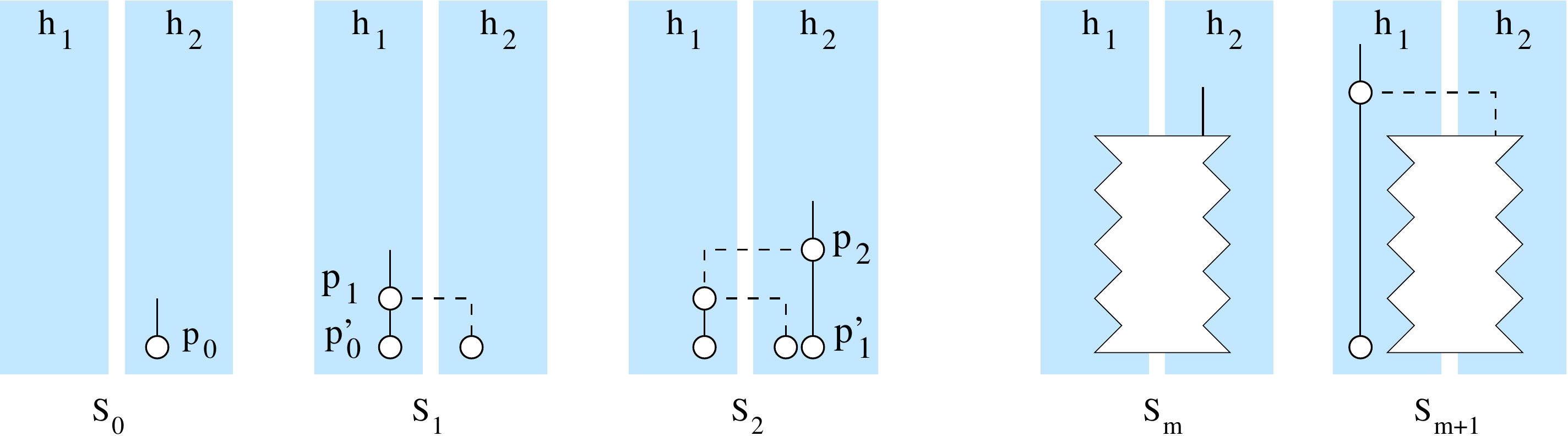}
\caption{Sewing trees $S_0$, $S_1$, $S_2$, and $S_{m+1}$ obtained from $S_m$.}\label{fig:sewing}
\end{figure}

First we introduce a gadget, called `sewing tree', that will help in the definition of our instance. A \emph{sewing tree} is a subtree of the parasite tree whose nodes are alternatively assigned to two host leaves $h_1$ and $h_2$ as follows. A single node $p_0$ with $\gamma(p_0)=h_2$ is a sewing tree $S_0$ of size $0$ and root $p_0$. Let $S_m$ be a sewing tree of size $m$ and root $p_m$ such that $\gamma(p_m)=h_2$ ($\gamma(p_m)=h_1$, respectively). In order to obtain $S_{m+1}$ we add a node $p_{m+1}$ with $\gamma(p_{m+1})=h_1$ ($\gamma(p_{m+1})=h_2$, respectively) and two children, $p_m$ and $p'_m$, with $\gamma(p'_m)=h_1$ ($\gamma(p'_m)=h_2$, respectively). See Fig.~\ref{fig:sewing} for examples of sewing trees. 
Intuitively, a sewing tree has the purpose of making costly from the point of view of the number of crossings the insertion of a host node $h_3$ between hosts $h_1$ and $h_2$, whenever $h_3$ contains several vertical arcs of $P$ towards leaves of the subtree rooted at $h_3$.

\begin{figure}[h]
\centering
\includegraphics[width=11cm]{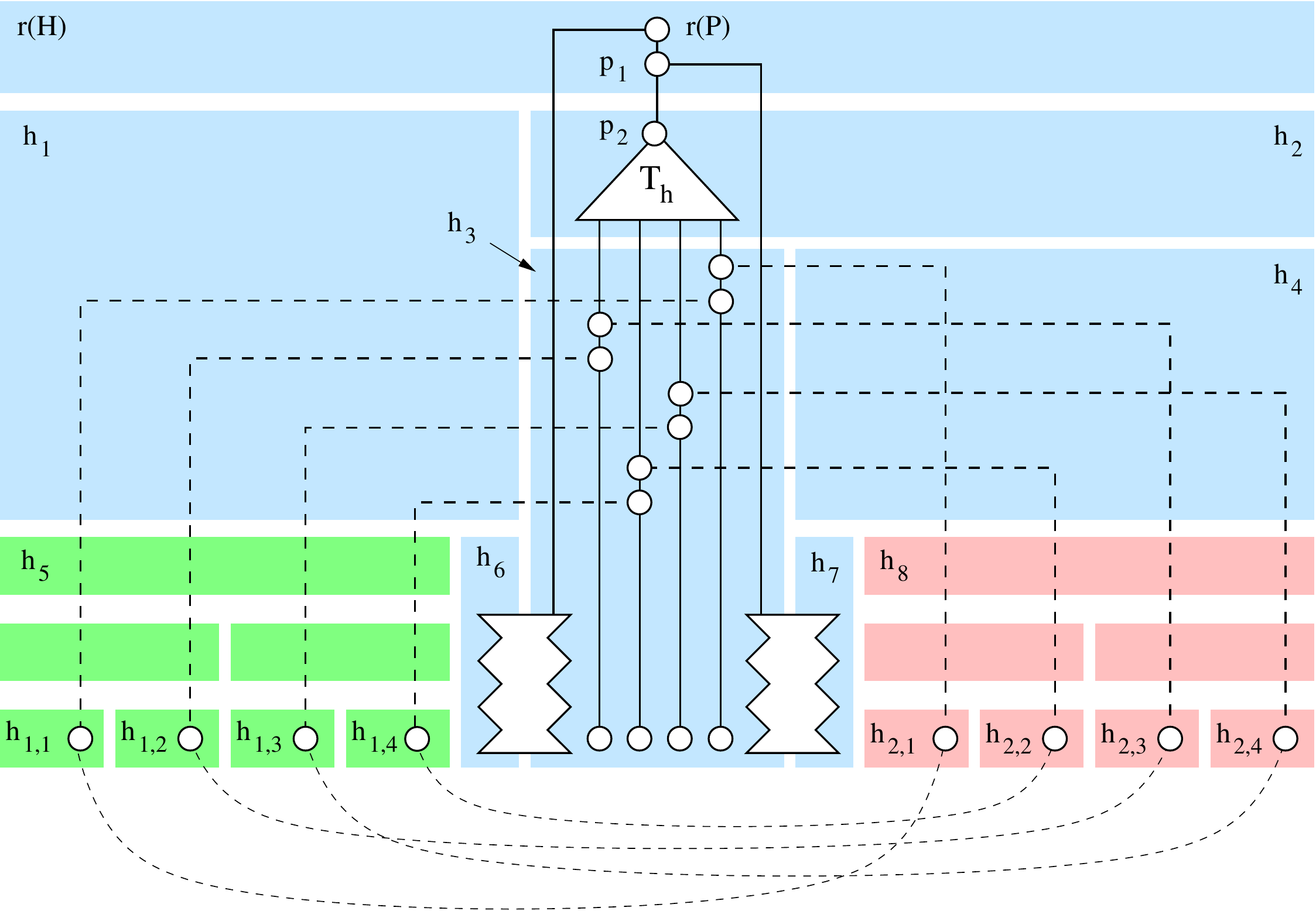}
\caption{
The construction of the instance of \textsc{RL} starting from an instance of \textsc{TTCM} with $h=2$. Filled green and pink are the subtrees of $H$ whose embeddings correspond to the embeddings of $T_1$ and $T_2$, respectively.}
\label{fig:reduction}
\end{figure}

Nodes $r(H), h_1, h_2, \dots, h_8$ of the host tree $H$ and their relationships are depicted in Fig.~\ref{fig:reduction}. Rooted at $h_5$ and $h_8$ we have two complete binary trees of height~$h$. Intuitively, these two subtrees of $H$ correspond to $T_1$ and $T_2$, respectively (they are drawn filled green and filled pink in Fig.~\ref{fig:reduction}). Hence, the leaves $l_{1,1}, l_{1,2}, \dots, l_{1,2^h}$ of $T_1$ are associated to the leaves $h_{1,1}, h_{1,2}, \dots, h_{1,2^h}$ of the subtree rooted at $h_5$, and, similarly, the leaves $l_{2,1}, l_{2,2}, \dots, l_{2,2^h}$ of $T_2$ are associated to the leaves $h_{2,1}, h_{2,2}, \dots, h_{2,2^h}$ of the subtree rooted at $h_8$. 

The root $r(P)$ of $P$ has $\gamma(r(P))=r(H)$. One child of $r(P)$ is the root of a sewing tree between $h_3$ and $h_6$. The other child $p_1$, with $\gamma(p_1)=r(H)$, has one child that is the root of a sewing tree between $h_3$ and $h_7$, and one child $p_2$, with $\gamma(p_2)=h_2$. Parasite $p_2$ is the root of a complete binary tree $T_h$ of height $h$, whose internal nodes are assigned to $h_2$, while the leaves are assigned to $h_3$.  
Each one of the $2^h$ leaves of $T_h$ is associated with a tangle of the instance $I_{\textsc{TTCM}}$. Namely, suppose $e=(l_{1,i},l_{2,j})$ is a tangle edge in the instance $I_{\textsc{TTCM}}$. Then, an arbitrary leaf $p_e$ of $T_h$ is associated with $e$. Node $p_e$ has children $p_{1,i}$, with $\gamma(p_{1,i})=h_{1,i}$, and $p'_e$, with $\gamma(p'_e)=h_3$. Node $p'_e$, in turn, has children $p_{2,j}$, with $\gamma(p_{2,j})=h_{2,j}$, and $p''_e$, with $\gamma(p''_e)=h_3$.
Finally, we pose $k'=k+2^h\cdot(2^h-1)$.

The proof concludes by showing that instance $I_{\textsc{TTCM}}$ is a yes instance of {\sc TTCM} if and only if instance $I_{\textsc{RL}}$ is a yes instance of {\sc RL} (see Appendix \ref{ap:theorem-complexity}).\qed
\end{sketch}

Since in the proof of Theorem~\ref{th:complexity} a key role is played by host-switch arcs, one could wonder whether an instance without host-switches is always planar. This is not the case: 
for any non-planar time-consistent reconciliation $\gamma \in \mathcal{R}(H, P,\varphi)$, there exists a time-consistent
reconciliation $\gamma_r \in \mathcal{R}(H, P,\varphi)$ that maps all internal nodes of $P$ to $r(H)$ and that
has no host-switch. If the absence of host-switches could guarantee planarity, $\gamma_r$ would be planar and, by Theorem~\ref{th:planar}, also $\gamma$ would be planar, leading to a contradiction.
%
Indeed, it is not difficult to construct (see Figure~\ref{fig:crossings}) reconciliations without host-switches and not planar.

\section{Heuristics for drawing reconciliations with few crossings}\label{se:heuristics}

Theorem~\ref{th:complexity} shows that a drawing of a reconciliation with the minimum number of crossings cannot be efficiently found. 
For this reason, we propose two heuristics aiming 
at producing HP-drawings with few crossings (Fig.~\ref{fig:compare} shows two examples of non-planar HP-drawings produced by the heuristics).
In the following we will briefly describe them.

\begin{figure}[tb]
\centering
\subfigure[]{\includegraphics[width=4.5cm]{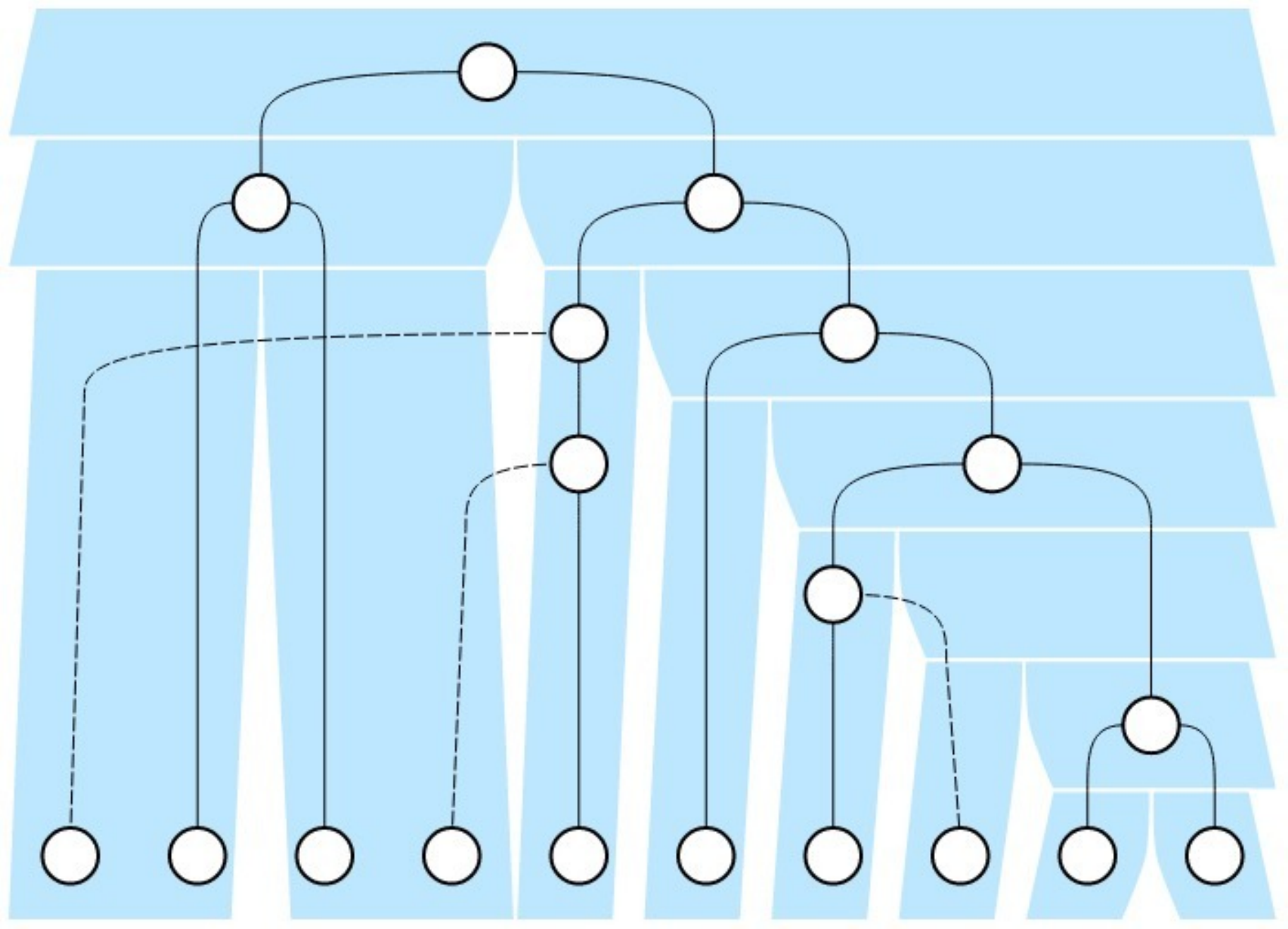}\label{fig:heur2}}
\hspace{1cm}
\subfigure[]{\includegraphics[width=4.5cm]{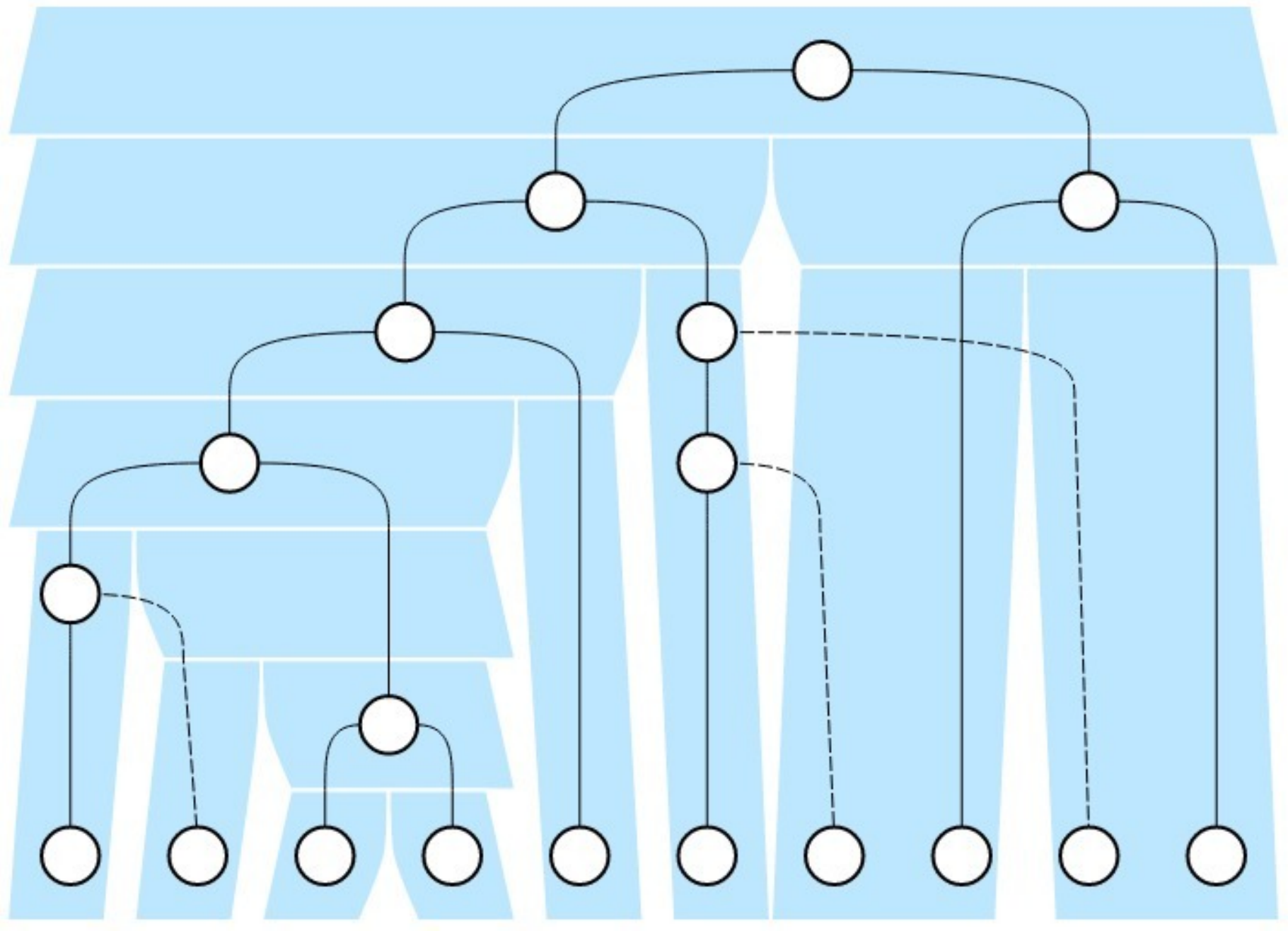}\label{fig:heur1}}
\caption{(1) An HP-drawing of a reconciliation of Gopher \& Lice drawn by \euristicaDue. (2) The same instance drawn by \euristicaUno.}\label{fig:compare}
\end{figure}

%
%
\subsection{Heuristic \euristicaDue}

This heuristic is based on the strategy of first drawing a large planar sub-instance and then adding non-planar arcs.
We hence construct a maximal planar subgraph $G_{pl}$ of tanglegram $\langle H, P, \varphi \rangle$
by adding to it one by one the following objects:
(i) all nodes of $H$ and of $P$; 
(ii) all arcs of $H$;
(iii) edge $(r(H), r(P))$; 
(iv) for each $l_p \in \mathcal{V}_L(P)$, edge $(l_p, \varphi(l_p))$;
(v) for each $p \in \mathcal{V}(P) \setminus \mathcal{V}_L(P)$, edge $(p, p')$, where $p'$ is any child of $p$ that is not a host-switch,
while the arc from $p$ to the sibling of $p'$
 is added to a set \texttt{missingArcs};
(vi) all arcs from \texttt{missingArcs} that is possible to add without introducing crossings
(all arcs that have not been inserted in $G_{pl}$ are stored in a set of \texttt{non-planarArcs}).  
A planar embedding of the graph $G_{pl}$ is used as input for Algorithm~\algoritmoPlanare so obtaining a planar drawing of part of reconciliation $\gamma$; arcs in \texttt{non-planarArcs} are added in a post-processing step. 

\remove{
\begin{verbatim}

MIGLIORIE: USO DI SKIPLIST:
	- al posto di provare arco per arco ogni volta, si prova a fare dei tentativi anche con più archi.
		- ogni sqrt(ArchiRimossi) si effettua la seguente prova:
			inserisco al posto di un solo arco sqrt(archiRimossi.length) archi.
			se il grafo rimane planare -> il salto è riuscito e quindi vado avanti con il nuovo grafo con tutti gli archi considerati.
			se il grafo non risulterà planare -> il salto non è riuscito e quindi si torna al grafo prima della prova del salto e si procede arco per arco.
		- con questo salto abbiamo che: 
			- se i salti vanno tutti a vuoto abbiamo fatto girare GDT un numero di volte pari a (archiRimossi.length + sqrt(archiRimossi.length));
			- se un salto va a segno archiRimossi.length
			- se tutti i salti vanno a segno sqrt(archiRimossi.length)


\end{verbatim}
}

%
%
\subsection{Heuristic \euristicaUno}

This heuristic is based on the observation that `long' host-switch arcs are more likely to cause crossings than `short' ones. Hence, this heuristic searches for an embedding of $H$ that reduces the distance between the end-nodes of host-switch arcs of $P$.
To do this, as a preliminary step, \euristicaUno chooses the embedding of $H$ with a preorder traversal as follows.
Let $v \in \mathcal{V}(H)$ be the current node of the traversal. Consider the set of nodes of $H$ that are ancestors or descendants of $v$. 
The removal of this set would leave two connected components, one on the left, denoted $V_{v,left}(H) \subseteq \mathcal{V}(H)$ and one on the right, denoted $V_{v,right}(H) \subseteq \mathcal{V}(H)$. Denote by $V_{v,left}(P)$ ($V_{v,right}(P)$, respectively) the set of parasite nodes mapped to some node in $V_{v,left}(H)$ ($V_{v,right}(H)$, respectively).
Moreover, denote by $V_v(P) \subseteq \mathcal{V}(P)$ the set of the parasite nodes mapped to the subtree of $H$ rooted at~$v$. 

%
%
If $v \in \mathcal{V}_L(H)$ no embedding choice has to be taken for $v$.
Otherwise let $v_1$ and $v_2$ be its children. 
For $i \in \{1,2\}$ and $X \in \{left,right\}$ compute the number $h_{v_i,X}$ of the host-switch arcs from $V_{v_i}(P)$ to $V_{v,X}(P)$ or vice versa. 
If $h_{1,right}+h_{2,left} > h_{2,right}+h_{1,left}$ then $v_1$ is embedded as the right child and $v_2$ as the left child of $v$, otherwise $v_2$ will be the right child and $v_1$ the left child.

Observe that the sets $V_{v,left}(P)$ and $V_{v,right}(P)$ can be efficiently computed while descending $H$. Namely, we start with $V_{r(H),left}(P)=V_{r(H),right}(P)=\emptyset$ and,  
supposing $v_{l}$ and $v_{r}$ are chosen to be the left and right children of $v$, respectively, we set 
$V_{v_{l},left}(P)=V_{v,left}(P)$, $V_{v_{l},right}(P)=V_{v,r}(P) \cup V_{v_{r}}$, 
$V_{v_{r},left}(P)=V_{v,left}(P) \cup V_{v_{l}}$, and $V_{v_{r},right}(P)=V_{v,right}(P)$.

It remains to describe how \euristicaUno places parasite nodes inside the representation of host nodes. First, we temporarily assign to each node $p \in \mathcal{V}(P)$ the lower $x-$ and $y-$coordinates inside $\gamma(p)$ (observe that all nodes mapped to the same host are overlapped).
For the leaves $\mathcal{V}(P)$ the temporary $y$-coordinate is definitive and only the $x$-coordinate has to be decided. 
We order the parasite leaves $p_1, p_2, \dots, p_k$ inside each host leaf $v_k$ as follows. We divide the leaves into two sets $L_{v, left}(P)$ and $L_{v,right}(P)$, where $L_{v, left}(P)$ contains the leaves associated with $v$ that have a parent with lower $x$-coordinates and $L_{v, right}(P)$ contains the remaining leaves associated with $v$. We order the set $L_{v, left}(P)$ ($L_{v, right}(P)$, respectively) ascending (descending, respectively) based on the $y$-coordinates of their parents. We place the set $L_{v, left}(P)$ and then the $L_{v,right}(P)$ inside $v$ according to their orderings. Once the leaves of $P$ have been placed, the remaining internal nodes of $P$ are placed according to the same algorithm used for planar instances by \algoritmoPlanare.


\section{Experimental evaluation}\label{se:experiments}


We collected standard co-phylogenetic tree instances from the domain literature. Table~\ref{ta:instances} shows their properties. 

\begin{table}
\centering
\begin{tabular} { | l | c | c | c | c | }
\hline
\textbf{Instance} & \textbf{~Acronym~} & \textbf{~\# hosts~} & \textbf{~\# par.~} & \textbf{~Planar~}\\ 
\hline\hline

\texttt{Caryophyllaceae \& Microbotryum}~\cite{CORE-ILP}   		& \texttt{CM} & 35 & 39 & No \\ \hline
\texttt{Stinkbugs \& Bacteria}~\cite{CORE-ILP}             		& \texttt{SB} & 27 & 23 & Yes\\ \hline

\texttt{Encyrtidae \& Coccidae}~\cite{Dal15}   				& \texttt{EC} & 13 & 19 & Yes \\ \hline
\texttt{Fishs \& Dactylogyrus}~\cite{Dal15}         			& \texttt{FD} & 39 & 101 & No \\ \hline
\texttt{Gopher \& Lice}~\cite{Dal15}                    		& \texttt{GL} & 15 & 19 & No \\ \hline
\texttt{Seabirds \& Chewing Lice}~\cite{Dal15}           	& \texttt{SC} & 21 & 27 & No \\ \hline
\texttt{Rodents \& Hantaviruses}~\cite{Dal15}       			& \texttt{RH} & 67 & 83 & No \\ \hline 
\texttt{Smut Fungi \& Caryophill. plants}~\cite{Dal15}   & \texttt{SFC} & 29 & 31 & No\\ \hline

\texttt{Pelican \& Lice (ML)}~\cite{Dal15}                	& \texttt{PML} & 35 & 35 & Yes \\ \hline
\texttt{Pelican \& Lice (MP)}~\cite{Dal15}                	& \texttt{PMP} & 35 & 35 & Yes \\ \hline
\texttt{Rodents \& Pinworms }~\cite{Dal15}                	& \texttt{RP} & 25 & 25 & No \\ \hline
\texttt{Primates \& Pinworms}~\cite{Dal15}                   & \texttt{PP} & 71 & 81 & No \\ \hline

\texttt{COG2085}~\cite{Dal15} 								& \texttt{COG2085} & 199 & 87 & No \\ \hline
\texttt{COG4965}~\cite{Dal15} 								& \texttt{COG4965} & 199 & 59 & No \\ \hline

\texttt{COG3715}~\cite{Dal15} 								& \texttt{COG3715} & 199 & 79 & No \\ \hline
\texttt{COG4964}~\cite{Dal15} 								& \texttt{COG4964} & 199 & 53 & No \\ \hline
\end{tabular}
%
\smallskip
\caption{The co-phylogenetic trees used to generate the datasuite.}\label{ta:instances}
\end{table}

Since reconciliations obtained from planar co-phylogenetic trees are always planar, we restricted our experiments to non-planar instances. 
In order to obtain a datasuite of reconciliations we used the Eucalypt tool~\cite{Dal15} to produce the set of minimum-cost reconciliations of each instance with costs 0, 2, 1, and 3 for co-speciation, duplication, loss, and host-switch, respectively. We configured the tool to filter out all time-inconsistent reconciliations based on the algorithm in~\cite{thl-sidlg-11}. Also, we bounded to 100 the reconciliations of each instance. 







We implemented the two heuristics \euristicaDue and \euristicaUno in JavaScript (but we used Python for accessing the file system and the GDToolkit library~\cite{dd-g-13} for testing planarity) and run the experiments on a Linux laptop with 7.7 GiB RAM and quadcore i5-4210U 1.70 GHz processor. 

\begin{table}
\centering

\begin{tabular}{ |c|c|c|c|c|c|c|c|c|c|c|c|c|c|}

    \cline{3-14}
	\multicolumn{2}{l}{} &
	\multicolumn{4}{|c|}{\euristicaUno} &
	\multicolumn{4}{|c|}{\texttt{SearchMaximalPlanar}$^*$} &
	\multicolumn{4}{|c|}{\euristicaDue} \\

    \cline{3-14}
	\multicolumn{2}{l}{} &
	\multicolumn{3}{|c|}{\textbf{\#Crossings}} &
	\multicolumn{1}{|c|}{\textbf{Avg}} &
	\multicolumn{3}{|c|}{\textbf{\#Crossings}} &
	\multicolumn{1}{|c|}{\textbf{Avg}} &
	\multicolumn{3}{|c|}{\textbf{\#Crossings}} &
	\multicolumn{1}{|c|}{\textbf{Avg}} \\

    \cline{1-5}
    \cline{7-9}
    \cline{11-13}

	\textbf{Inst.} & \textbf{\#Rec.} & 
	\textbf{Max} & \textbf{Min} & \textbf{Avg} & \textbf{ms} &
	\textbf{Max} & \textbf{Min} & \textbf{Avg} & \textbf{ms} &
	\textbf{Max} & \textbf{Min} & \textbf{Avg} & \textbf{ms} \\
\hline
\hline

\texttt{CM} 	& 64 & 30 & 15 & 21 & 0.5 
                & 21 & 13 & 17 & 644 
                & 20 & 10 & 16 & 485 
                \\ \hline

\texttt{FD} 	& 80 & 84 & 55 & 69 & 1 
                & 108 & 74 & 92 & 7289 
                & 110 & 67 & 91 & 4596 
                \\ \hline
\texttt{GL} 	& 2 & 1 & 1 & 1 & 0 & 1 & 1 & 1 & 180 & 2 & 2 & 2 & 67 \\ \hline
\texttt{PP} 	& 72 & 6 & 2 & 3 & 1 & 4 & 3 & 3 & 4840 & 2 & 1 & 1 & 1154 \\ \hline
\texttt{RH}		& 100 & 11 & 11 & 11 & 2 & 11 & 9 & 10 & 1710 & 15 & 10 & 12 & 1701 \\ \hline 
\texttt{RP} 	& 3 & 4 & 2 & 3 & 1 & 3 & 3 & 3 & 737 & 3 & 3 & 3 & 195 \\ \hline
\texttt{SC} 	& 1 & 6 & 6 & 6 & 0 & 4 & 4 & 4 & 499 & 4 & 4 & 4 & 166 \\ \hline
\texttt{SFC} 	& 16 & 22 & 11 & 17 & 0 & 16 & 12 & 13 & 412 & 20 & 11 & 15 & 355 \\ \hline 

\texttt{COG2085} & 100 & 80 & 58 & 70 & 7 & 95 & 84 & 89 & 20540 & 99 & 68 & 82 & 17270 \\ \hline
\texttt{COG4965} & 100 & 125 & 79 & 97 & 8 & 68 & 57 & 60 & 9901 & 65 & 52 & 58 & 5636 \\ \hline


\end{tabular}    
\smallskip
\caption{The results of the experiments.}\label{ta:results}
\end{table}

Table~\ref{ta:results} shows the results of the experiments. Planar instances \texttt{SB}, \texttt{EC}, \texttt{PML}, and \texttt{PMP} were not used to generate reconciliations. Also, instances \texttt{COG3715} and \texttt{COG3715} did not produce any time-consistent reconciliation. 
For all the other phylogenetic-trees, the second column of Table~\ref{ta:results} shows the number of reconciliations computed by Eucalypt (we bounded to 100 the reconciliations of \texttt{RH}, \texttt{COG2085}, and \texttt{COG4965}). Table~\ref{ta:results} is vertically divided into three sections, each devoted to a different heuristics. 
Each section shows the minimum, maximum, and average number of crossings
and the average computation time for the HP-drawings produced by the heuristics on the reconciliations obtained for the phylogenetic-tree specified in the first column. 

The section labeled \texttt{SearchMaximalPlanar}$^*$ shows the results of \euristicaDue where we computed the embedding of tree $H$ (which is the most expensive algorithmic step) once for all the reconciliations of the same instance. Hence, differently from the other two sections, the computation times reported in this section refer to the sum of computation times for all the reconciliations obtained from the same instance. 

From Table~\ref{ta:results} it appears that heuristic \euristicaDue is much slower than \euristicaUno. This could have been predicted, since \euristicaDue runs a planarity test several times. However, the gain in terms of crossings is questionable. Although there are instances where \euristicaDue appears to outperform \euristicaUno (for example, \texttt{CM}, \texttt{PP}) this is hardly a general trend. We conclude that aiming at planarity is not the right strategy for minimizing crossings in this particular application context. 

The strategy of computing the embedding of $H$ once for all reconciliations of the same co-phylogenetic tree (central section labeled \texttt{SearchMaximalPlanar}$^*$ of Table~\ref{ta:results}) seems to be extremely effective in reducing computation times. For example, on instance \texttt{COG2085}, where this heuristics needed 20.5 seconds, it actually used 205 msec per reconciliation, about 11\% of the time needed by \euristicaDue, at the cost of very few additional crossings.



\section{Conclusions and Future Work}

This paper introduces a new and intriguing simultaneous visualization problem, i.e.\ producing readable drawings of the  reconciliations of co-phylogenetic trees. Also, a new metaphor is proposed that takes advantage both of the space-filling and of the node-link visualization paradigms. We believe that such a hybrid strategy could be effective for the simultaneous visualization needs of several application domains. 

As future work, we would like to address the problem of visually exploring and analyzing sets of reconciliations of the same co-phylogenetic tree, which is precisely the task that several researchers in the biological field need to perform. 
Heuristic \texttt{SearchMaximalPlanar}$^*$ is a first step in this direction, since it maintains the mental map of the user by fixing the drawing of~$H$.
Finally, we would like to adapt heuristics for the reduction of the crossings of tanglegram drawings, such as those in~\cite{bhtw-ffpad-09,nvwh-dbtee-09,szh-trptn-11}, to our problem and we would like to perform user tests to assess the effectiveness of the proposed metaphor.

\subsection*{Acknowledgments}

We thank Riccardo Paparozzi for first experiments on the visualization of co-phylogenetic trees. 
Moreover, we are grateful to Marie-France Sagot and Blerina Sinaimeri for proposing us the problem and for the interesting discussions.

\bibliographystyle{splncs03}
\bibliography{bibliography}


\newpage
\appendix

\newgeometry{left=2cm,right=2cm,top=2cm,bottom=2cm}

\section{Formal Definitions Omitted from Section~\ref{se:preliminaries}}\label{ap:defs}

%
%
\subsection{Definition of Full Rooted Binary Tree}\label{ap:tree-defs}

A \emph{rooted tree} $T=(V,A)$ is a set ${\cal V}(T)$ of $n$ nodes and a set ${\cal A}(T)$ of $n-1$ directed arcs $(u,v)$, with $u,v \in {\cal V}(T)$, such that: (i) each node $v$ is the endpoint of exactly one arc with the exception of one node called the \emph{root} of~$T$ and denoted $r(T)$ and (ii) there is a direct path from $r(T)$ to any other node $n \in \mathcal{V}(T)$. 
The \emph{depth} $d(v)$ of a node $v \in \mathcal{V}(T)$ is the length (number of edges) of the direct path from $r(T)$ to $v$.
The tree is said to be \emph{binary} if each node $u$ has at most two outgoing arcs $(u,v_1)$ and $(u,v_2)$, and is said to be \emph{full binary} if each node has either two or zero outgoing arcs. Nodes with two outgoing arcs are said to be \emph{internal} and their set is denoted ${\cal V}_I(T)$, while nodes with zero outgoing arcs are said to be \emph{leaves} and their set is denoted as~${\cal V}_L(T)$, so being ${\cal V}(T)={\cal V}_I(T) \cup {\cal V}_L(T)$. An \emph{ancestor} of a node $u$ is any node on the path from $r(T)$ to $u$ (including $r(T)$ and $u$). 

%
%
\subsection{Definition of Co-Evolutionary Events}\label{ap:events}


In a reconciliation, four types of events may take place~\cite{C98,PC98} (refer to Fig.~\ref{fig:events}): 

\begin{description}

\item[Co-speciation.] When, for the two children $p_1$ and $p_2$ of $p$, $\gamma(p_1)$ and $\gamma(p_2)$ are incomparable and $lca(\gamma(p_1),\gamma(p_2)) = \gamma(p)$. Intuitively, both parasite and host speciate.

\item[Duplication.] When for the two children $p_1$ and $p_2$ of $p$, $lca(\gamma(p_1),\gamma(p_2)) \in \{\gamma(p_1), \gamma(p_2)\}$, that is, the children of $p$ are mapped to comparable nodes. Intuitively, the parasite speciates but the host does not.

\item[Loss.] Whenever an arc $(p_i,p_j)$ corresponds to a path of length $k$ from $\gamma(p_i)$ to $\gamma(p_j)$, then we have $k-1$ losses in the intermediate nodes of $H$. Intuitively, we have a loss each time the host speciates but the parasite does not. 

\item[Host-switch.] When, for an arc $(p_i,p_j)$, $lca(\gamma(p_i),\gamma(p_j)) \neq \gamma(p_i)$. Intuitively, a host-switch is the unlikely event that a child $p_j$ of parasite $p_i$ is transferred to a host that is not a descendant of $\gamma(p_i)$. Due to Property (3) of reconciliations, a parasite $p_i$ may be the source of at most one host-switch arc.
\end{description}

\begin{figure}
\centering
\includegraphics[width=6cm]{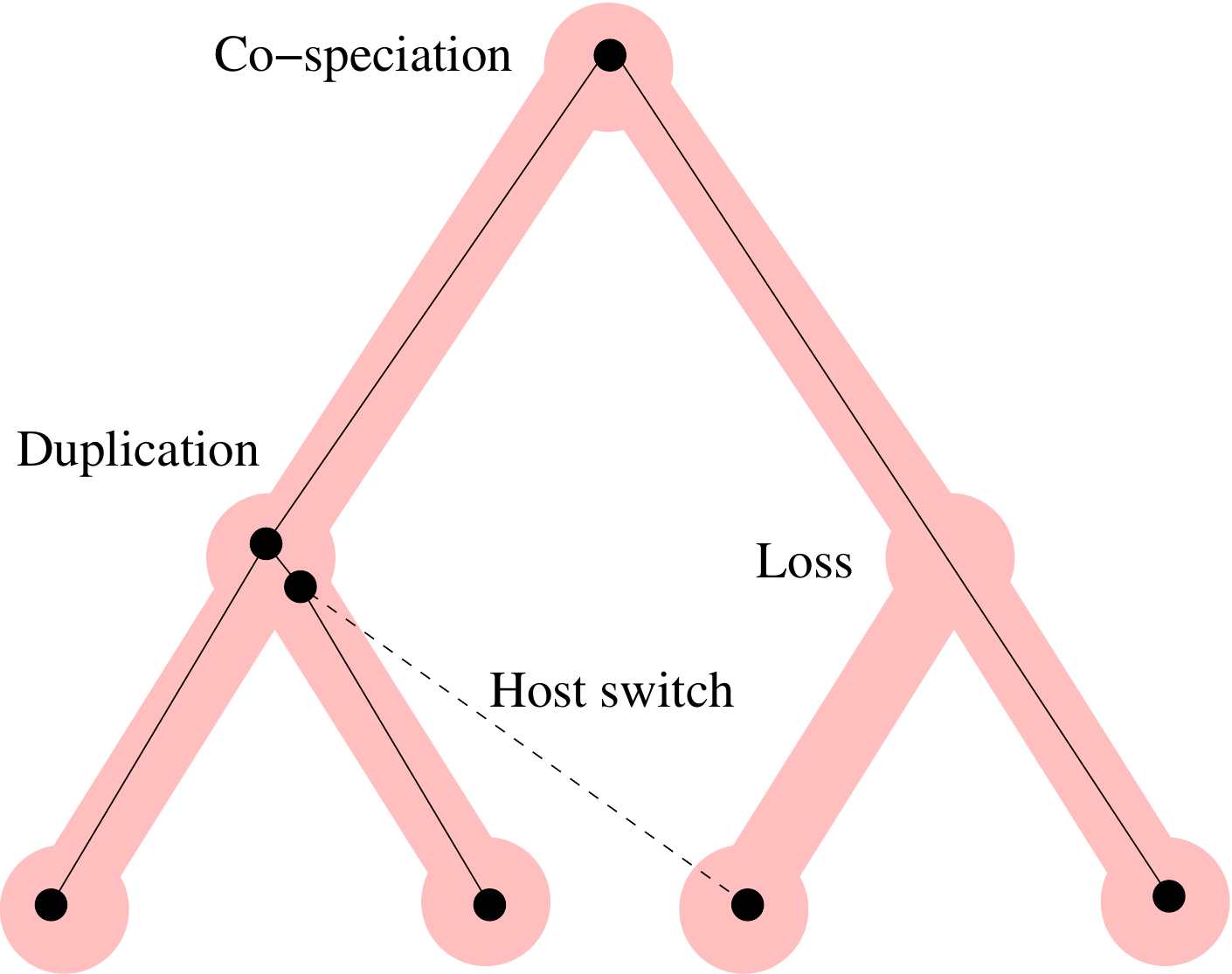} 
\caption{Co-evolutionary events in a reconciliation of a co-phylogenetic tree.}\label{fig:events}
\end{figure}

%
%

\section{Co-Phylogenetic Trees Visualization Tools}\label{ap:tools}

This section contains examples of representations of reconciliations obtained with state-of-the-art tools. Figure~\ref{fig:core-pa} shows an example of a straight-line representation obtained with CoRe-PA~\cite{CORE-ILP}, adopting the strategy of representing both $H$ and $P$ with a traditional node-link metaphor. As it can be noted, the diagram tends to be cluttered when several symbionts are associated with the same host, and there is sometimes ambiguity on what host is associated with what symbiont. 

\begin{figure}[htb]
\centering
\subfigure[]{\includegraphics[width=8cm]{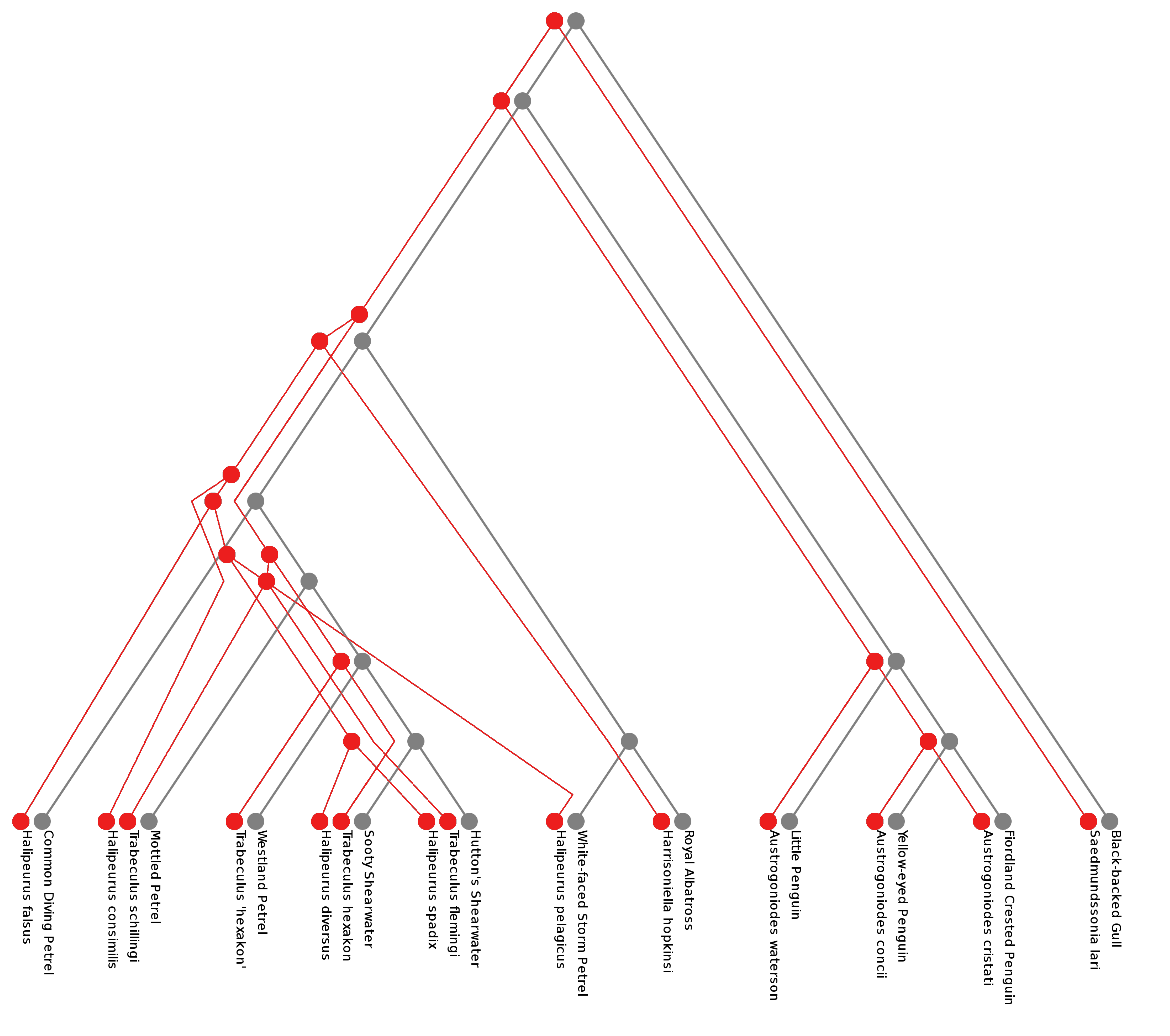}\label{fig:core-pa}}
\subfigure[]{\includegraphics[width=8cm]{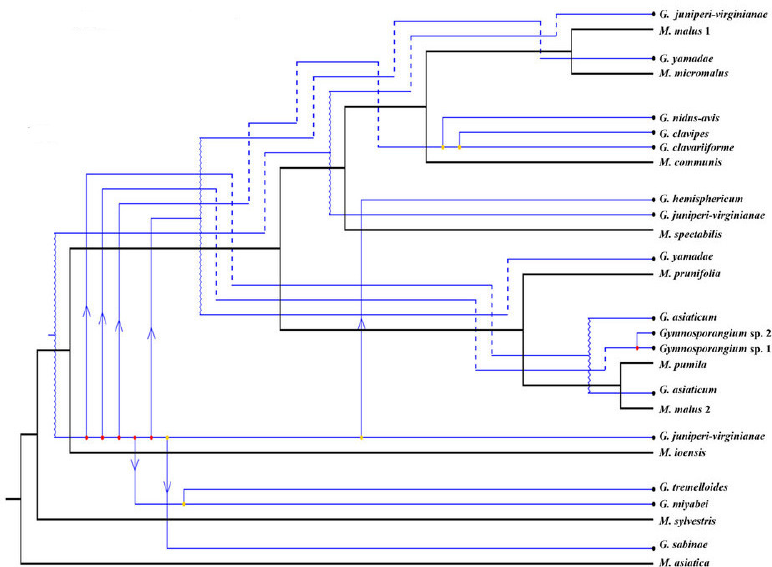}\label{fig:jane-4}}
\caption{\subref{fig:core-pa} Visualization of a reconciliation of Seabirds \& Chewing Lice co-evolution trees obtained with CoRe-PA. \subref{fig:jane-4} Cophylogenetic analysis of the Gymnosporangium-Malus pathosystem conducted with Jane 4~\cite{jane4}.}
\end{figure}

Figure~\ref{fig:jane-4} shows a representation of a reconciliation obtained with Jane 4~\cite{jane4}. The picture is taken from~\cite{zflc-ipsgs-16}. Both $H$ and $P$ are represented with the node-link metaphor but in this case an orthogonal representation has been chosen to avoid overlapping of arcs and nodes.


{\em CophyTrees}~\cite{eucalypt}, the viewer associated with the Eucalypt tool~\cite{Dal15} keeps the idea of representing the edges of $H$ as tubes and the edges of $P$ as lines (see Fig.~\ref{fig:eucalypt}). In the representation it is sometimes unclear when more than one node of $P$ is mapped on a single node of $H$.   

\begin{figure}[htb]
\centering
\subfigure[]{\includegraphics[width=8cm]{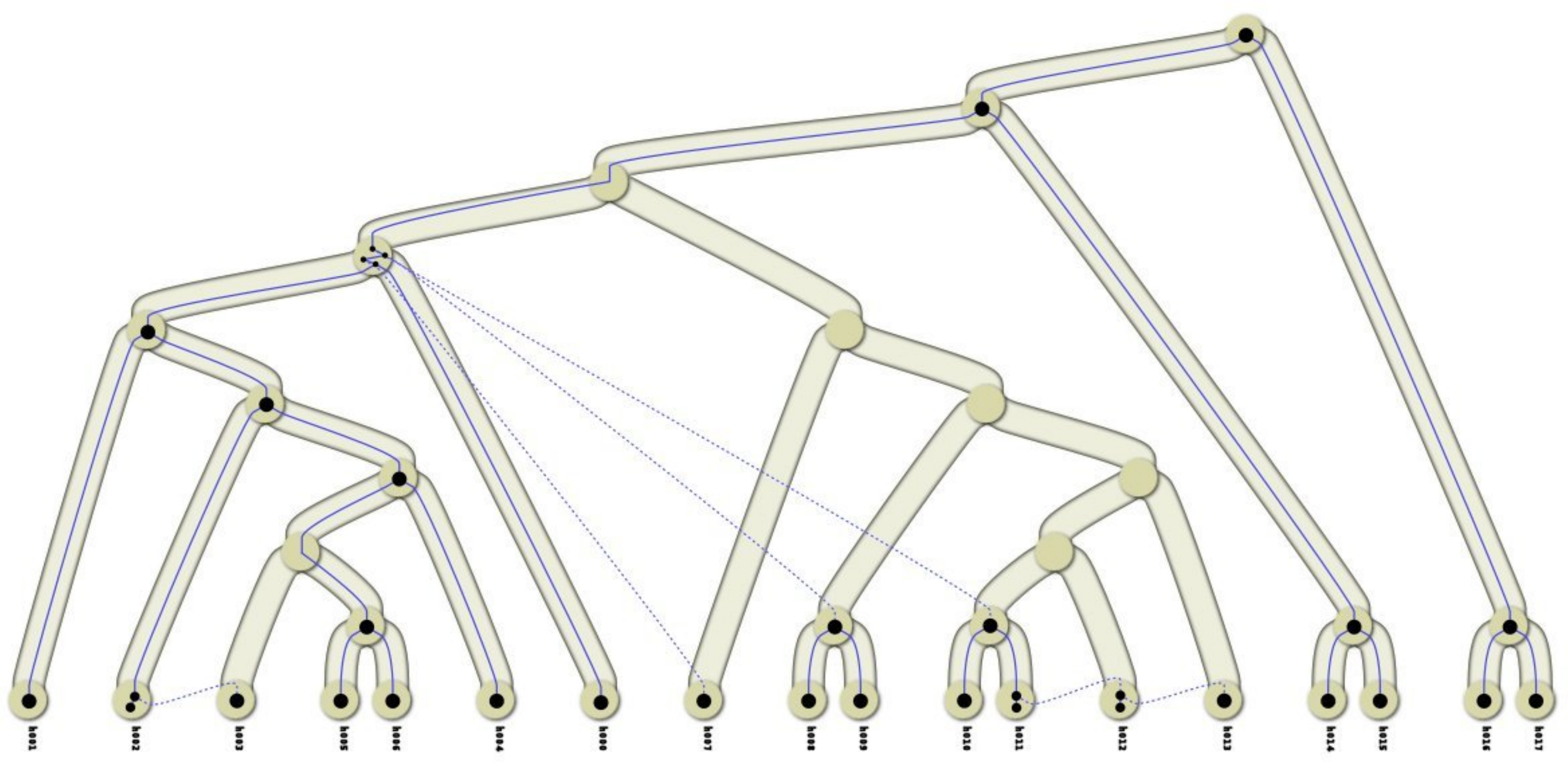}\label{fig:eucalypt}}
\subfigure[]{\includegraphics[width=8cm]{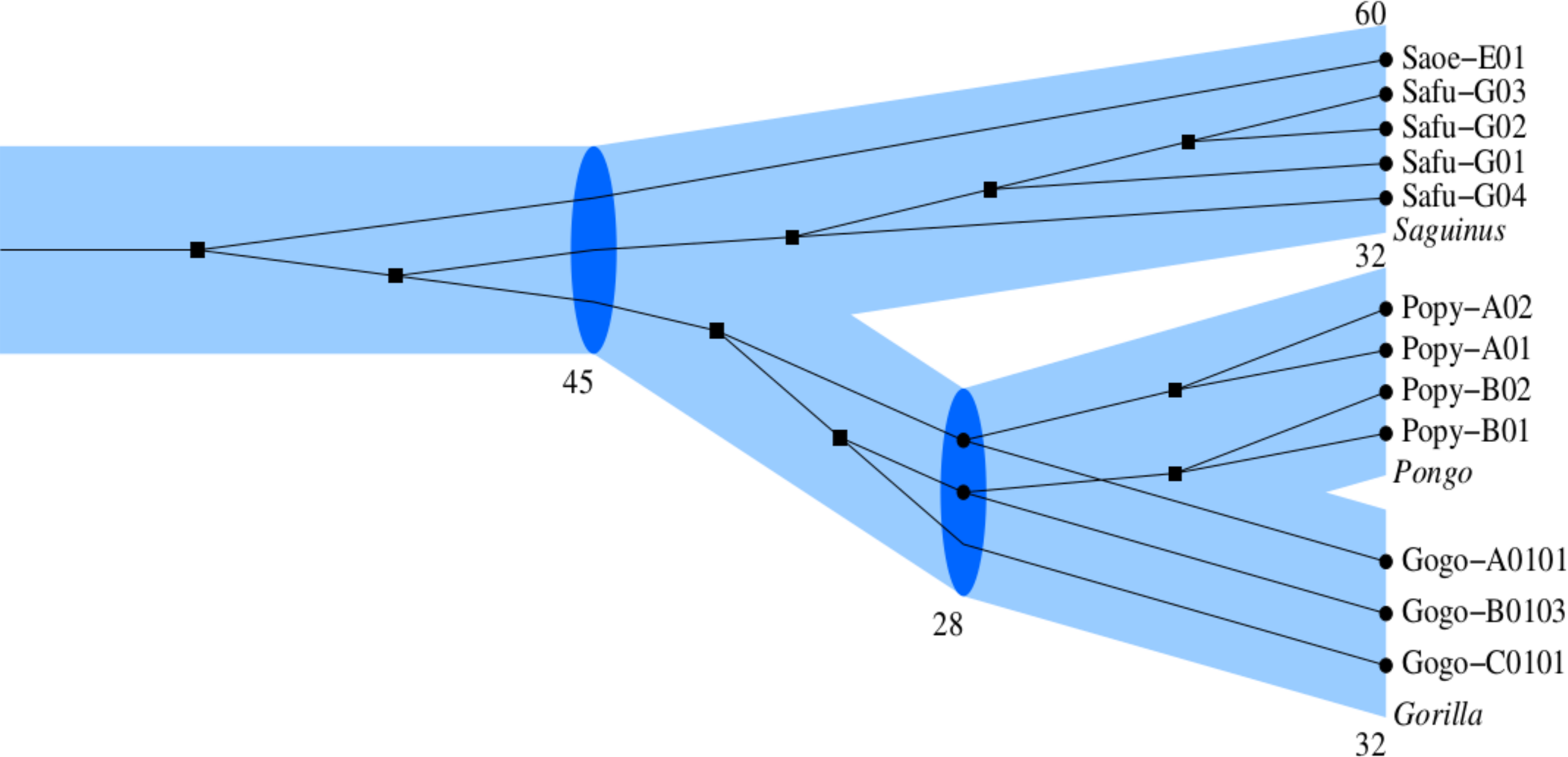}\label{fig:primetv}}
\caption{\subref{fig:eucalypt} An example of co-phylogenetic tree drawn by {\em CophyTrees}~\cite{eucalypt}. \subref{fig:primetv} Primetv-generated illustration of the reconciled tree showing the evolution of the gene family Major Histocompatibility Complex class I in Gorilla, Orangutan, and Tamarin (picture from~\cite{primetv}).}
\end{figure}


{\em Primetv} \cite{primetv} represents $H$ as a tree whose arcs are tubes and nodes are ellipses, while $P$ is represented inside the pipes of $H$ (ignoring the ellipses that represent nodes of $H$). Arcs of $P$ are straight segments. It is easy to see that, when $H$ and $P$ have a large number of nodes, it becomes impossible to clearly attribute parasites to hosts (see Fig.~\ref{fig:primetv}).


{\em SylvX} \cite{Sylvx} is a more complex reconciliation viewer which implements classical phylogenetic graphic operators (swapping, highlighting, etc.) and methods to ease interpretation and comparison of reconciliations (multiple maps, moving, shrinking sub-reconciliations). $H$ is represented as an orthogonal drawing where arcs are again represented as pipes, while $P$ is embedded inside $H$. Growing the size of the instance the drawing quickly becomes cluttered (see Fig. \ref{fig:sylvx}).

\begin{figure}[htb]
\centering
\includegraphics[width=10cm]{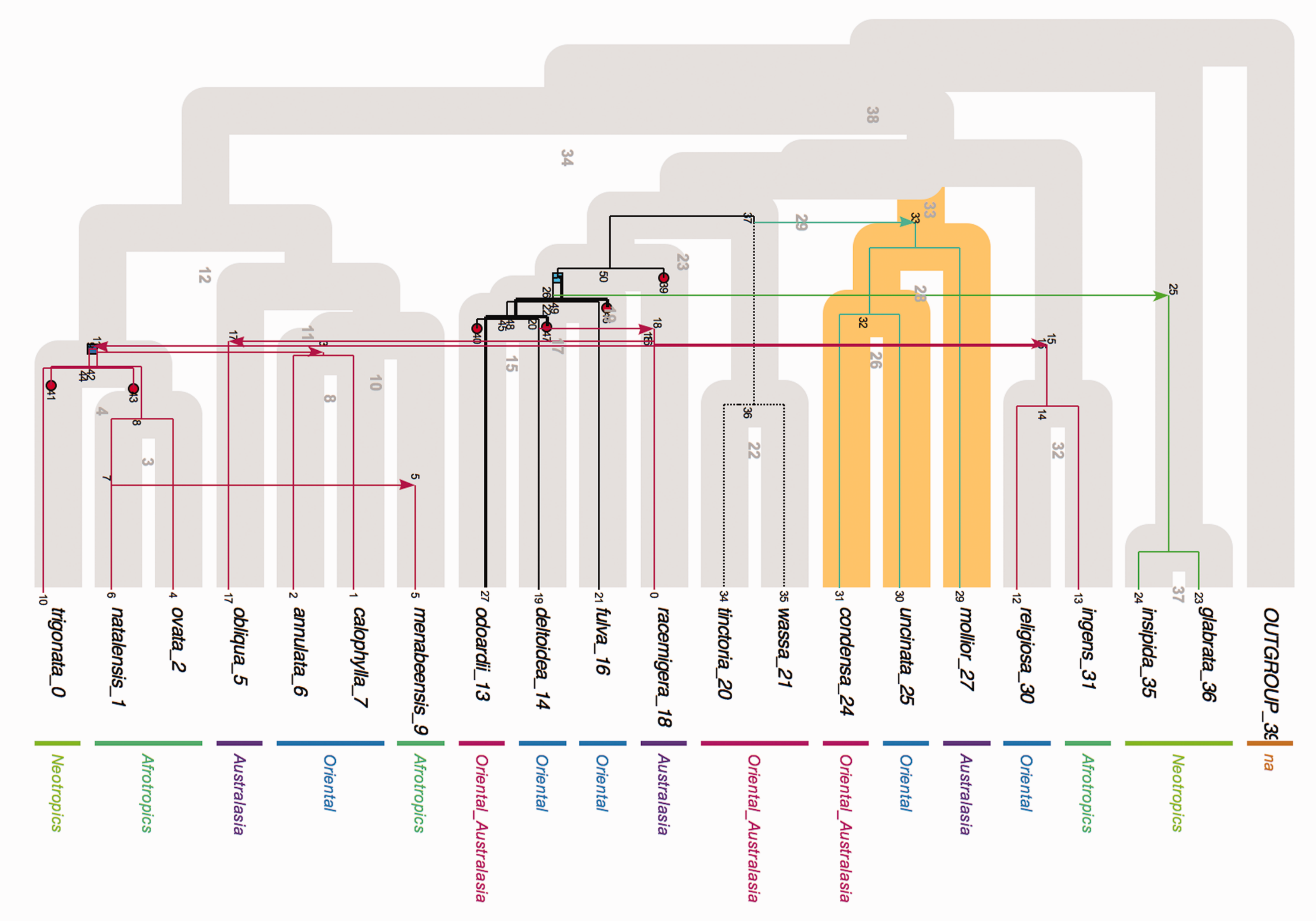}
\caption{Example of visualization of the host/parasite dataset ficus/fig wasp obtained with SylvX (picture from~\cite{Sylvx}).}\label{fig:sylvx}
\end{figure}

\begin{figure}[h]
\centering
\includegraphics[width=15cm]{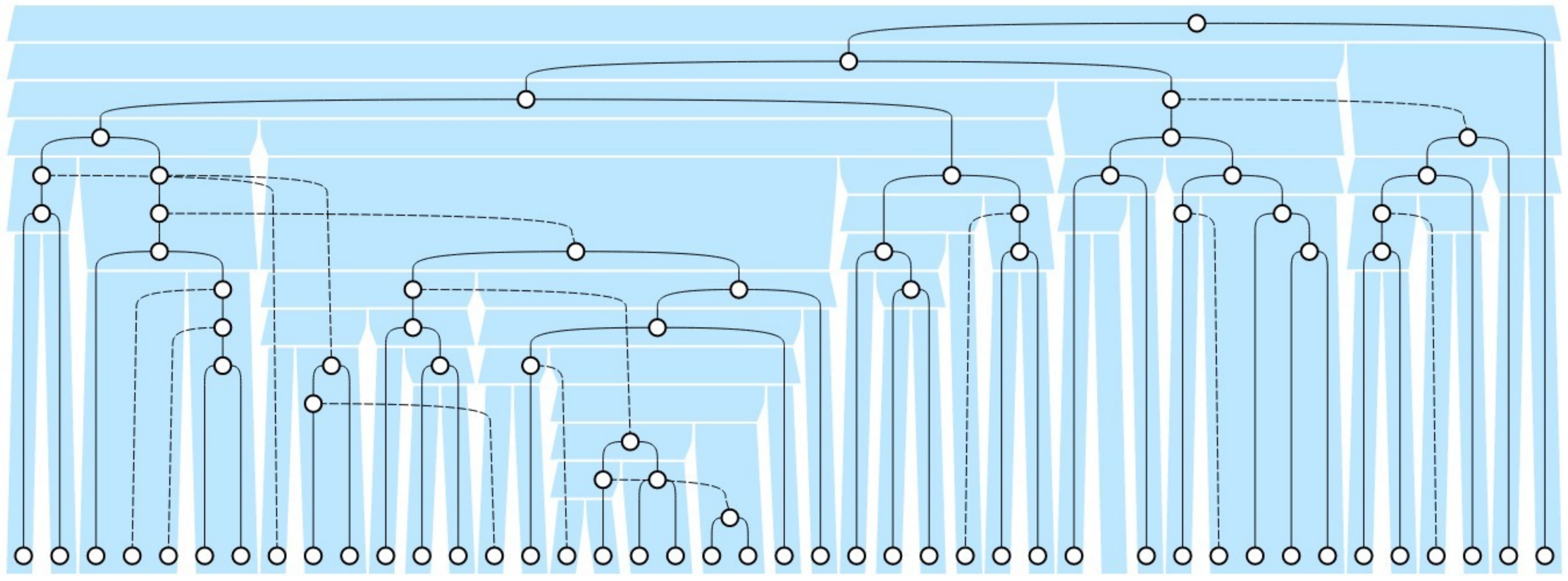}
\caption{An HP-drawing of a reconciliation of the co-phylogenetic tree of Rodents \& Hantaviruses produced by heuristic \euristicaUno.}
\label{fig:RH}
\end{figure}

%

\section{Full Proof of Theorem~\ref{th:planar}}\label{ap:theorem-planar}

\theoremplanar*
\begin{proof}
First, we prove that $(2)$ implies $(1)$.
Consider a planar drawing $\Gamma(\gamma)$ of $\gamma \in \mathcal{R}(H,P, \varphi)$ and let $l$ be the horizontal line passing through the bottom border of $\Gamma(\gamma)$. Observe that the leaves of $P$ lie above $l$.
Construct a tanglegram drawing $\Delta$ of $\langle H,P, \varphi\rangle$ as follows:
\begin{inparaenum}[(a)]%
\item Draw $H$ by placing each node $h \in \mathcal{V}(H)$ in the center of the rectangle representing $h$ in $\Gamma(\gamma)$ and by representing each arc $a \in \mathcal{A}(H)$ as a suitable curve between its incident nodes;
\item draw $P$ in $\Delta$ as a mirrored drawing with respect to $l$ of the drawing of $P$ in $\Gamma(\gamma)$;
\item connect each leaf $p \in L(P)$ to the host $\gamma(p)$ with a straight-line segment.
\end{inparaenum}
It is immediate that $\Delta$ is a tanglegram drawing of $\langle H,S,\varphi\rangle$ and that it is planar whenever $\Gamma(\gamma)$ is. 


Proving that $(1)$ implies $(2)$ is more laborious. 
Let $\Delta$ be a planar tanglegram drawing of $\langle H,P, \varphi\rangle$ (Fig.~\ref{fig:tanglegram}(a)). We construct a drawing $\Gamma(\gamma)$ of the given time-consistent reconciliation $\gamma \in \mathcal{R}(H,S, \varphi)$ as follows.
%
First, insert into the arcs of $P$ dummy nodes of degree two to represent losses, obtaining a new tree $P'$ (Fig.~\ref{fig:tanglegram}(b)) as follows. Consider an arc $(p,q) \in \mathcal{A}(P)$ that corresponds to a path $h_1, h_2, \dots, h_k$ of length $k$ from $\gamma(p)$ to $\gamma(q)$ in $H$. Then for $i=2,\dots, k-1$ repeatedly insert parasite $p_i$ between $p_{i-1}$ and $q$, where $p_1=p$, and set $\gamma(p_i)=h_i$. 

Since $\gamma$ is time-consistent, consider any ordering $\pi'$ of $\mathcal{V}(P')$ consistent with $H$. Remove from $\pi'$ the leaves of $P$ and renumber the remaining nodes obtaining a new ordering $\pi$ from $1$ to $|\mathcal{V}(P')-\mathcal{V}_L(P')|$. 
%
\begin{description}
\item[Regarding $y$-coordinates:] all the leaves of $P'$ have $y$-coordinate $1$, that is, they are placed at the bottom of the drawing, while each internal node $p \in \mathcal{V}(P')\setminus\mathcal{V}_L(P')$ has $y$-coordinate $2\pi(p)+1$ (see Fig.~\ref{fig:downward}(a)).
\item[Regarding $x$-coordinates:] each leaf $p \in \mathcal{V}_L(P)$ has $x$-coordinate $2{\sigma}(p)+1$, where ${\sigma}(p)$ is the left-to-right order of the leaves of $T_2$ in $\Delta$. The $x$-coordinate of an internal node $p$ of $P$ is copied from one of its children $p_1$ or $p_2$, arbitrarily chosen if none of them is connected by a host-switch, the one (always present) that is not connected by a host-switch otherwise. 
\end{description}

Let $h$ be a node of $\mathcal{V}(H)$; 
rectangle $R_h$, representing $h$ in $\Gamma$, has the minimum width that is sufficient to span all the parasites contained in the subtree $T_h(H)$ of $H$ rooted at $h$ (hence, it spans the interval $[x_{min}-1,x_{max}+1]$, where $x_{min}$ and $x_{max}$ are the minimum and maximum $x$-coordinates of a parasite contained in $T_h(H)$, respectively). The top border of $R_h$ has $y$-coordinate $y_{\textsc{min}}-1$, where $y_{\textsc{min}}$ is the minimum $y$-coordinate of a parasite node contained in the parent of $h$. The bottom border of $R_h$ is $y_{min}-1$, where $y_{min}$ is the minimum $y$-coordinate of a parasite node contained in $h$ (see Fig.~\ref{fig:downward}(b)).
%

Now, we show that the obtained representation $\Gamma(\gamma)$ is planar and downward. 
First observe that two rectangles $R_i$ and $R_j$ cannot overlap. In fact, by construction $R_i$ and $R_j$ can overlap with their $x$-coordinates only if the corresponding hosts $h_i$ and $h_j$ have a descendant in common, which is ruled out if $h_i$ and $h_j$ are not comparable. If $h_i$ and $h_j$ are comparable, then $R_i$ and $R_j$ surely overlap with their $x$-coordinates but, by construction, they cannot overlap with their $y$-coordinates. 

Since the $y$-coordinate of the parasites was assigned based on a consistent ordering $\pi$, all arcs of $P$ are drawn downward. 
Finally, the representation of $P$ is planar as the embedding of $P$ mirrors the embedding of $T_2$ and the drawing of $P$ is downward.
\qed
\end{proof}

\begin{figure}[h]
\centering
\begin{tabular}{c @{\hspace{2em}} c @{\hspace{2em}} c}
\includegraphics[width=4.5cm]{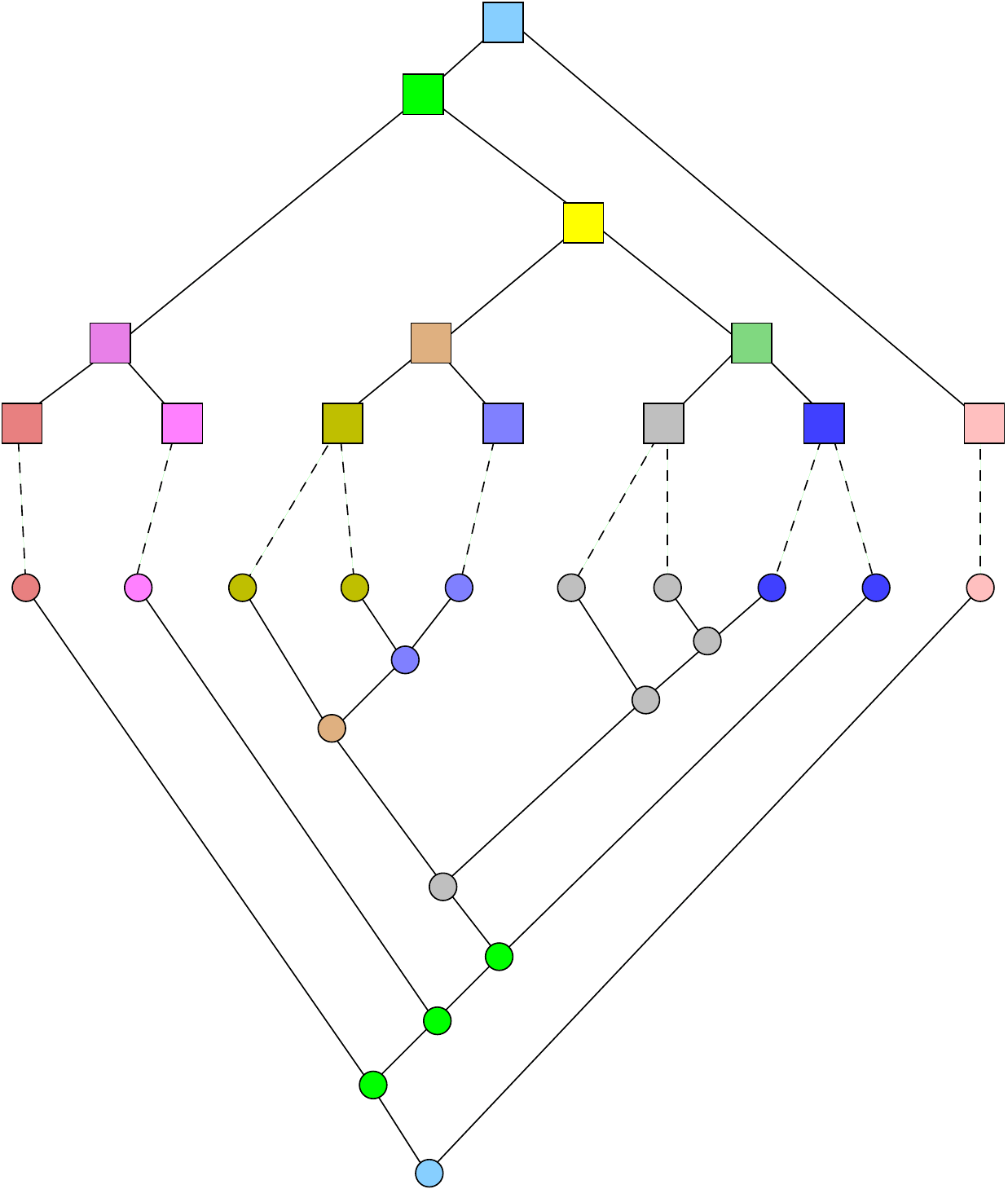}~~~~ &
~~~~\includegraphics[width=4.5cm]{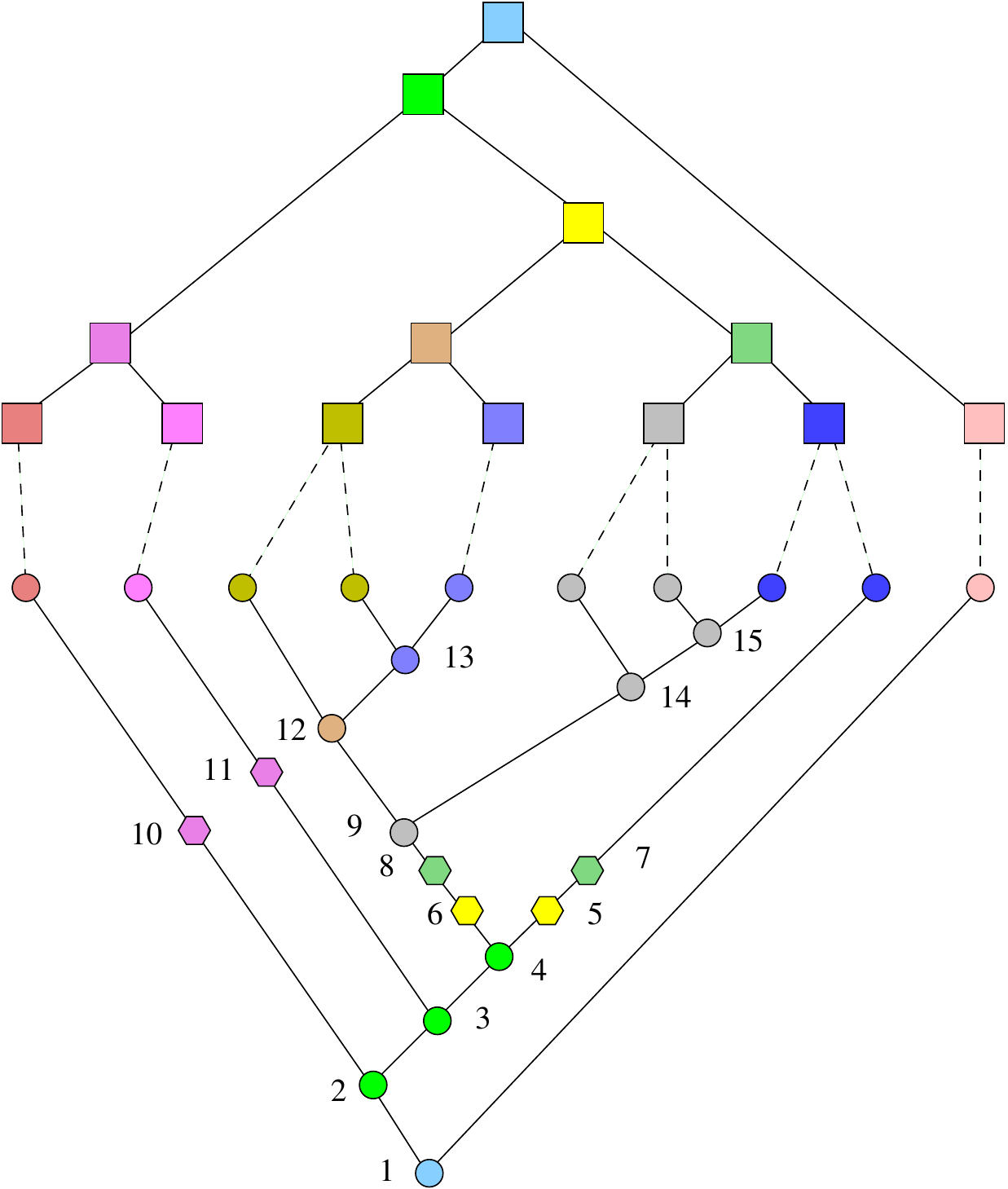} & \\
(a) & (b) 
\end{tabular}
\caption{(a) A planar tanglegram drawing of an instance $\langle H,P,\varphi\rangle$ of co-phylogenetic tree (Encyrtidae \& Coccidae). The colors of the nodes represent a reconciliation $\gamma$.
(b) The tree $P$ has been enriched with degree-two nodes (the hexagonal-shaped ones) accounting for losses.}\label{fig:tanglegram}
\end{figure}

\begin{figure}[h]
\centering
\begin{tabular}{c @{\hspace{2em}} c @{\hspace{2em}} c}
\includegraphics[width=4.5cm]{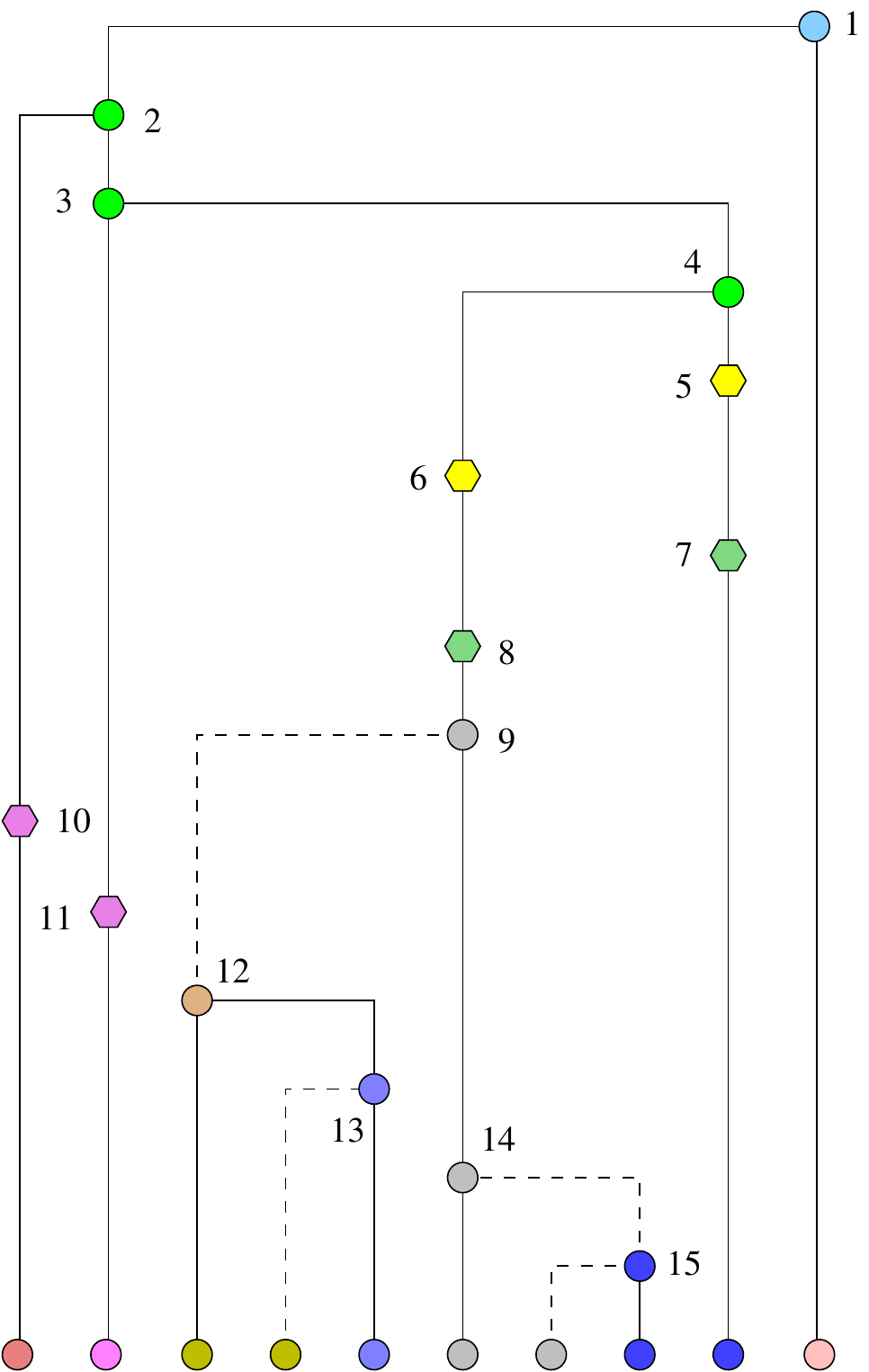} &
\includegraphics[width=4.5cm]{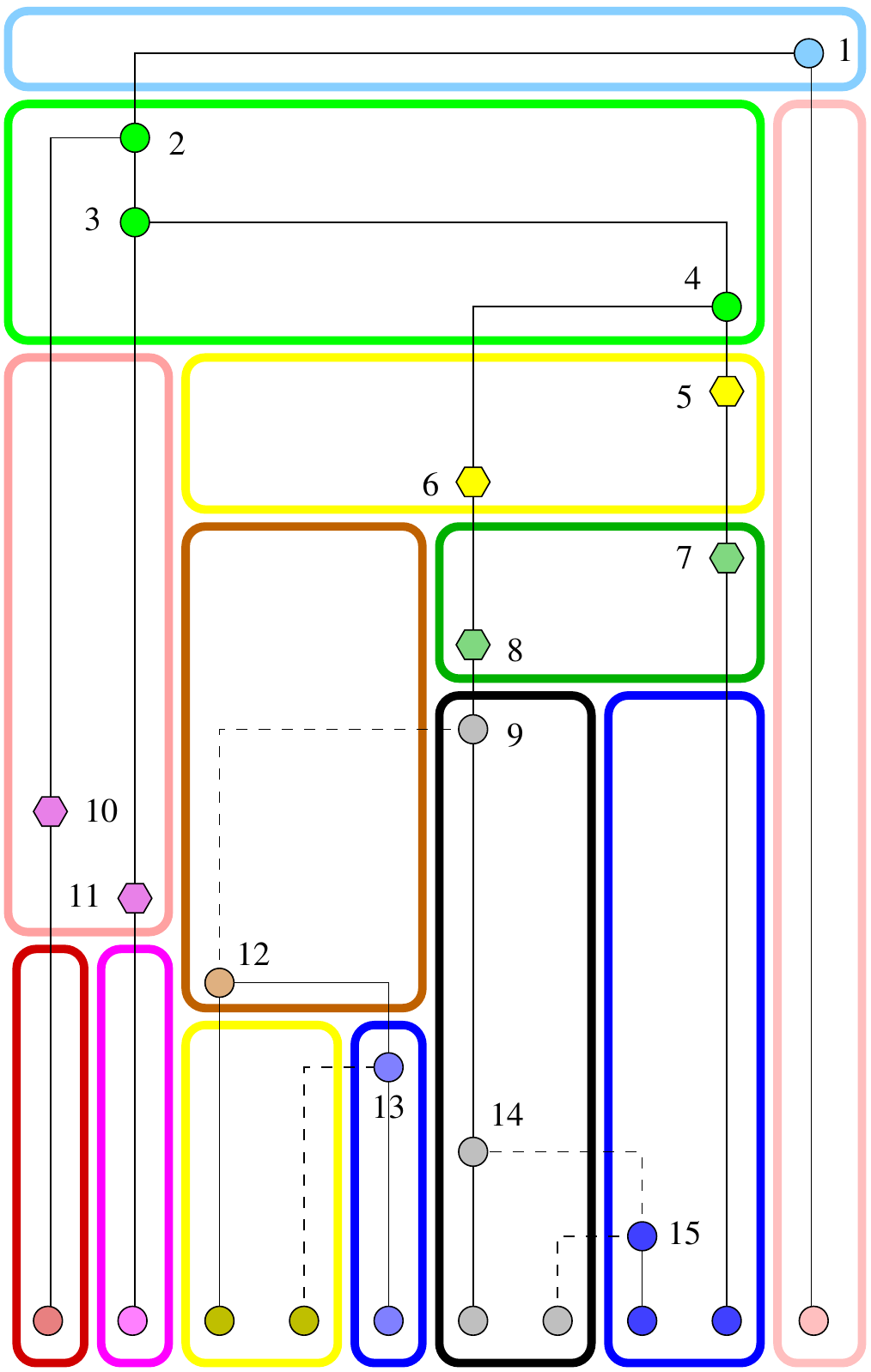} & \\
(a) & (b) 
\end{tabular}
\caption{(a) The drawing of tree $P$ obtained as described in the proof of Theorem~\ref{th:planar}. (b) The drawing of $\gamma$ obtained as described in the proof of Theorem~\ref{th:planar}.}\label{fig:downward}
\end{figure}

\section{Full Proof of Theorem~\ref{th:complexity}}\label{ap:theorem-complexity}

\theoremcomplexity*

\begin{proof}
Problem \textsc{RL} is in NP since we can non-deterministically explore all possible HP-drawings of $\gamma$ inside an area that has maximum width $|\mathcal{V}_L(P)| + |\mathcal{V}_L(H)|$ and maximum height $|\mathcal{V}(P)| + |\mathcal{V}(H)|$. A drawing is defined by assigning $x$- and $y$-coordinates to all nodes of $P$ and the coordinates of the top-left and bottom-right vertices of each rectangle associated to a node of $H$. Once coordinates have been non-deterministically assigned it remains to check if the obtained drawing is an HP-drawing of $\gamma$ and if the number of crossings is at most~$k$. Both these tasks can be performed in polynomial time.

Let $I_{\textsc{TTCM}} = \langle T_1, T_2, \psi, k \rangle$ be an instance of \textsc{TTCM}, where $T_1$ and $T_2$ are complete binary trees of height $h$, $\psi$ is a one-to-one mapping between $\mathcal{V}_L(T_1)$ and $\mathcal{V}_L(T_2)$, and $k$ is a constant. We show how to build an instance $I_{\textsc{RL}} =\langle \gamma \in \mathcal{R}(H, S,\varphi), k' \rangle$ of \textsc{RL}.

First we introduce a gadget, called `sewing tree', that will help in the definition of our instance. A \emph{sewing tree} is a subtree of the parasite tree whose nodes are alternatively assigned to two host leaves $h_1$ and $h_2$ in the following way. A single node $p_0$ with $\gamma(p_0)=h_2$ is a sewing tree $S_0$ of size $0$ and root $p_0$. Let $S_m$ be a sewing tree of size $m$ and root $p_m$ such that $\gamma(p_m)=h_2$ ($\gamma(p_m)=h_1$, respectively). In order to obtain $S_{m+1}$ we add a node $p_{m+1}$ with $\gamma(p_{m+1})=h_1$ ($\gamma(p_{m+1})=h_2$, respectively) and two children, $p_m$ and $p'_m$, with $\gamma(p'_m)=h_1$ ($\gamma(p'_m)=h_2$, respectively). See Fig.~\ref{fig:sewing} for examples of sewing trees. 
Intuitively, a sewing tree has the purpose of making costly the insertion of a host $h_3$ between hosts $h_1$ and $h_2$, whenever $h_3$ contains several vertical arcs of $P$ towards leaves of $h_3$.

Host tree $H$ has a root $r(H)$ with two children $h_1$ and $h_2$. Host $h_1$ has two children $h_5$ and $h_6$, which is a leaf. Host $h_2$ has a leaf child $h_3$ and a child $h_4$. In turn, $h_4$ has a leaf child $h_7$ and a child $h_8$. Finally, $h_5$ and $h_8$ are the roots of two binary trees of depth $h$. 
Intuitively, $T_1$ corresponds in $H$ to the subtree rooted at $h_5$ (filled green in Fig.~\ref{fig:reduction}) while $T_2$ corresponds in $H$ to the subtree rooted at $h_8$ (filled pink in Fig.~\ref{fig:reduction}). Hence, the leaves $l_{1,1}, l_{1,2}, \dots, l_{1,2^h}$ of $T_1$ are associated to the leaves $h_{1,1}, h_{1,2}, \dots, h_{1,2^h}$ of the subtree rooted at $h_5$, and, similarly, the leaves $l_{2,1}, l_{2,2}, \dots, l_{2,2^h}$ of $T_2$ are associated to the leaves $h_{2,1}, h_{2,2}, \dots, h_{2,2^h}$ of the subtree rooted at $h_8$. 

Let $k'=k+2^h\cdot(2^h-1)$.
The root $r(P)$ of $P$ has $\gamma(r(P))=r(H)$. One child of $r(P)$ is the root of a sewing tree of size $k'+1$ between $h_3$ and $h_6$. The other child $p_1$, with $\gamma(p_1)=r(H)$, has one child that is the root of a sewing tree of size $k'+1$ between $h_3$ and $h_7$, and one child $p_2$, with $\gamma(p_2)=h_2$. Parasite $p_2$ is the root of a complete binary tree $T_h$ of height $h$, whose internal nodes are assigned to $h_2$, while the leaves are assigned to $h_3$.  
Each one of the $2^h$ leaves of $T_h$ is associated with a tangle of the instance $I_{\textsc{TTCM}}$. Namely, suppose $e=(l_{1,i},l_{2,j})$ is a tangle in the instance $I_{\textsc{TTCM}}$. Then, an arbitrary leaf $p_e$ of $T_h$ is associated with $e$. Node $p_e$ has children $p_{1,i}$, with $\gamma(p_{1,i})=h_{1,i}$, and $p'_e$, with $\gamma(p'_e)=h_3$. Node $p'_e$, in turn, has children $p_{2,j}$, with $\gamma(p_{2,j})=h_{2,j}$, and $p''_e$, with $\gamma(p''_e)=h_3$.

Now we show that instance $I_{\textsc{TTCM}}$ is a yes instance of {\sc TTCM} if and only if the instance $I_{\textsc{RL}}$ is a yes instance of {\sc RL}.

Suppose $\langle T_1, T_2, \psi, k \rangle$ admits a tanglegram drawing with at most $k$ crossings. We show how to build a drawing of $I_{\textsc{RL}}$ with at most $k'$ crossings. We draw $r(H), h_1, h_2, \dots, h_8$ according to Fig.~\ref{fig:reduction} and we embed the subtrees of $H$ rooted at $h_5$ and $h_8$ as the tree $T_1$ and $T_2$, respectively. 
Now consider two tangles $e_1$ and $e_2$ that do not cross in the tanglegram drawing of the co-phylogenetic tree $\langle T_1, T_2, \psi\rangle$. In the instance $I_{\textsc{RL}}$ they correspond to two subtrees rooted at two leaves $p_{e_1}$ and $p_{e_2}$ of $T_h$ that can be drawn in the HP-drawing  introducing only two crossings (refer to Fig.~\ref{fig:reduction-non-crossing}).
Conversely, if two tangles cross in the tanglegram drawing of $\langle T_1, T_2, \psi\rangle$, the corresponding subtrees in the instance $I_{\textsc{RL}}$ introduce necessarily three crossings.  
Since the number of tangles is $2^h$, they can be paired in $\frac{2^h\cdot(2^h-1)}{2}$ ways.
Hence, the total number of crossings is $k+2\frac{2^h\cdot(2^h-1)}{2}=k+2^h\cdot(2^h-1)=k'$ 

\begin{figure}[h]
\centering
\includegraphics[width=12cm]{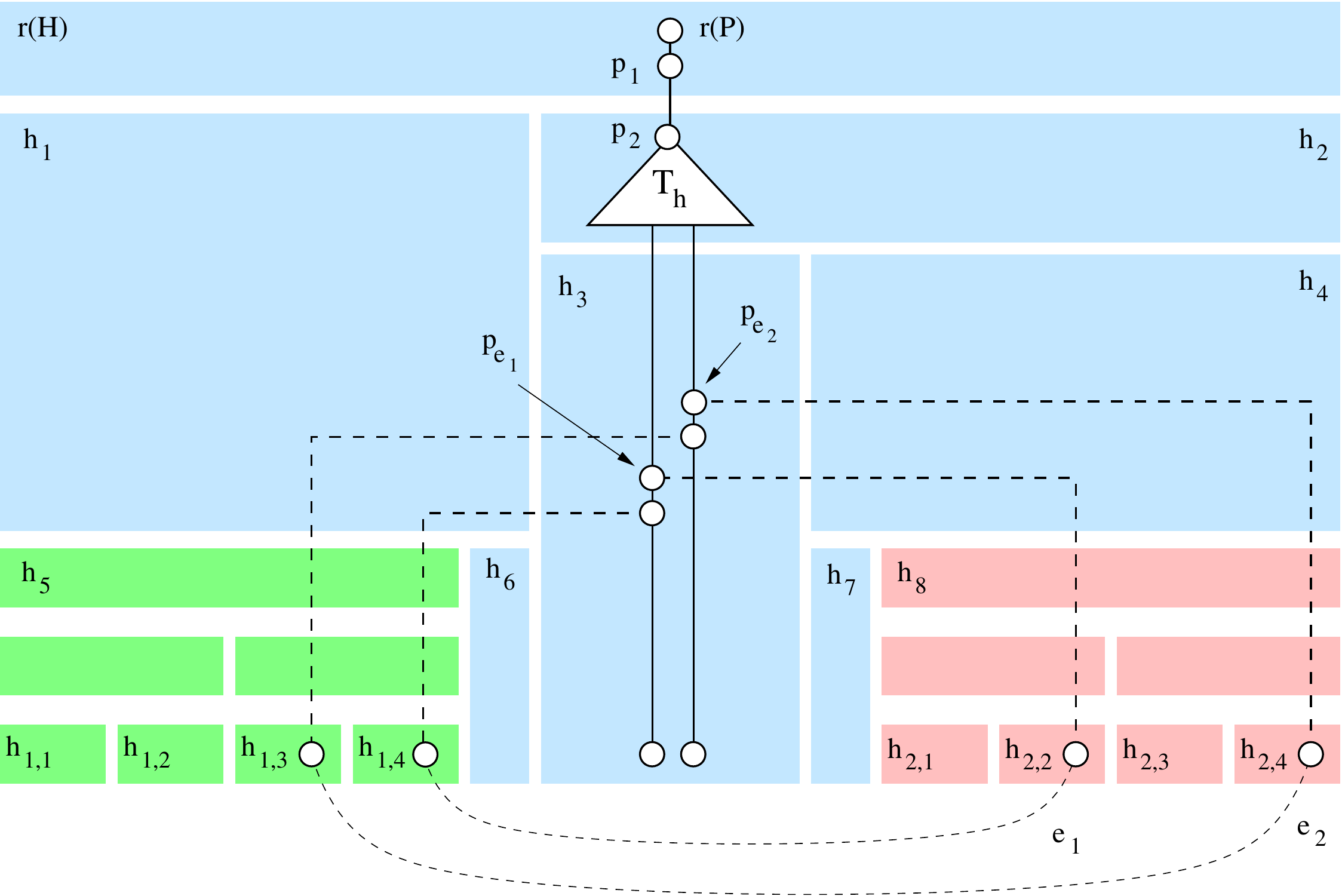}
\caption{ 
When two tangles $e_1$ and $e_2$ do not cross in the tanglegram drawing of $I_{\textsc{TTCM}}$  we have two crossings in the HP-drawing of $I_{\textsc{RL}}$.}
\label{fig:reduction-non-crossing}
\end{figure}

Conversely, suppose the instance $I_{\textsc{RL}}$ admits a drawing with at most $k'$ crossings. We show that the original instance $I_{\textsc{TTCM}}$ admits a tanglegram drawing with at most $k$ crossings. 
An immediate consequence of the sewing tree of size $k'+1$ between $h_6$ and $h_3$ is that any HP-drawing of the reconciliation $\gamma$ such that $h_6$ and $h_3$ are not adjacent has more than $k'$ crossings. Analogously, any HP-drawing of $\gamma$ where $h_7$ is not adjacent to $h_3$ has more than $k'$ crossings.
If follows that in any HP-drawing of $\gamma$ with at most $k'$ crossings the embedding of hosts $r(H), h_1, h_2, \dots, h_8$ is exactly the one represented in Fig.~\ref{fig:reduction}, up to a horizontal flip. 
The proof is concluded by showing that a tanglegram drawing of $\langle T_1, T_2, \psi\rangle$ with $k$ crossings can be constructed by embedding $T_1$ as the subtree of $H$ rooted at $h_5$ and by embedding $T_2$ as the subtree rooted at $h_8$.   
\qed
\end{proof}

\begin{figure}[h]
\centering
\includegraphics[width=7cm]{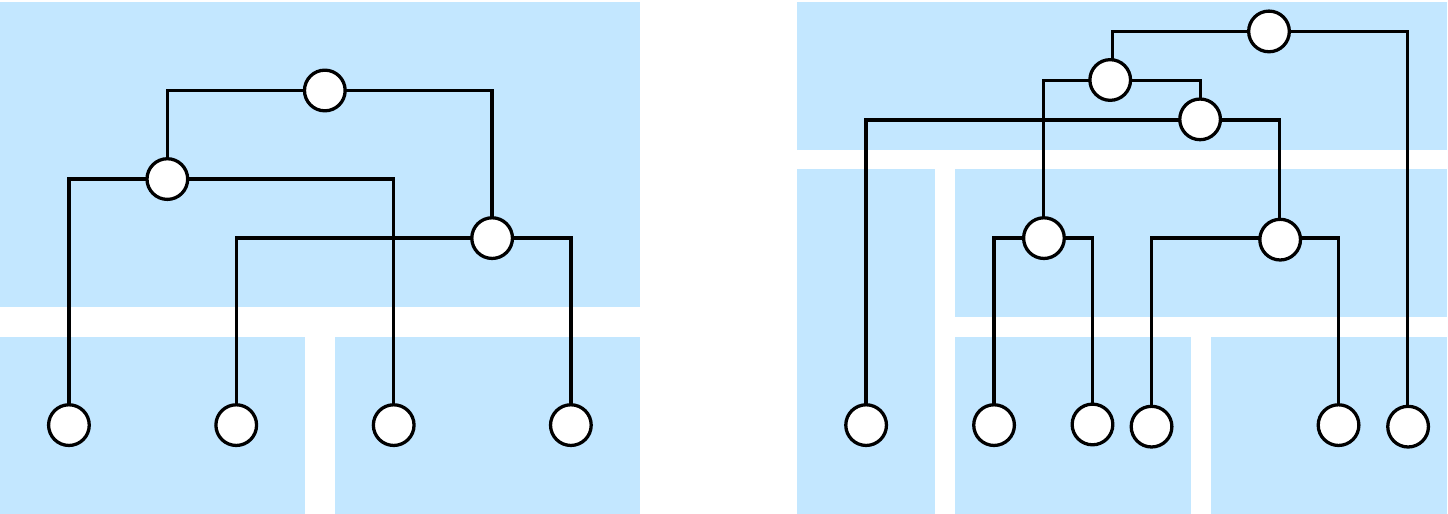}
\caption{Configurations with unavoidable crossings.}\label{fig:crossings}
\end{figure}


\end{document}